\documentclass[]{siamltex1213}
\usepackage[utf8]{inputenc}
\usepackage{amsmath}
\usepackage{amsfonts}
\usepackage{amssymb}
\usepackage{tabularx}
\usepackage{multirow}

\newcommand{\R}{\mathbb{R}}%
\newtheorem{lem}{Lemma}
\newtheorem{rem}{Remark}

\begin{document}

\title{Development and Analysis of a Block-Preconditioner for the Phase-Field Crystal Equation}
\author{Simon Praetorius\footnotemark[2]\ \footnotemark[3]
\and Axel Voigt\footnotemark[3]}

\maketitle

\begin{abstract}
We develop a preconditioner for the linear system arising from a finite element discretization of the Phase Field Crystal (PFC) equation. The PFC model serves as an atomic description of crystalline materials on diffusive time scales and thus offers the opportunity to study long time behaviour of materials with atomic details. This requires adaptive time stepping and efficient time discretization schemes, for which we use an embedded Rosenbrock scheme. To resolve spatial scales of practical relevance, parallel algorithms are also required, which scale to large numbers of processors. The developed preconditioner provides such a tool. It is based on an approximate factorization of the system matrix and can be implemented efficiently. The preconditioner is analyzed in detail and shown to speed up the computation drastically.
\end{abstract}


 
\renewcommand{\thefootnote}{\fnsymbol{footnote}}

\footnotetext[2]{Corresponding author: \url{simon.praetorius@tu-dresden.de} (Simon Praetorius)}
\footnotetext[3]{Institut f{\"u}r Wissenschaftliches Rechnen, Technische Universit{\"a}t Dresden, D-01062 Dresden, Germany}

\renewcommand{\thefootnote}{\arabic{footnote}}

\section{Introduction}
The Phase Field Crystal (PFC) model was introduced as a phenomenological model for solid state phenomena on an atomic scale \cite{Elder2002,Elder2004}. However, it can also be motivated and derived through classical dynamic density functional theory (DDFT) \cite{Elder2007,Teeffelen2009} and has been used for various applications in condensed and soft matter physics, see the review \cite{Emmerich2013} and the references therein. Applications include non-equilibrium processes in complex fluids \cite{Archer2009, Praetorius2011}, dislocation dynamics \cite{Chan2010}, nucleation processes \cite{Backofen2007, Backofen2010a, Guo2012, Backofen2014}, (dendritic) growth  \cite{Fallah2013, Wu2012, Tang20111} and grain growth \cite{Backofen2014a}.

The main solution methods for the PFC model, which is a non-linear 6th order parabolic partial differential equation, are finite-difference discretizations and spectral methods, which are combined with an explicit or semi-implicit time-discretization. Numerical details are described in  \cite{Cheng2008, Hirouchi2009, Hu2009, Tegze2009, Elsey2013}.

Recently, the PFC model has been coupled to other field variables, such as flow \cite{Praetorius2011}, orientational order \cite{Achim2011, Praetorius2013} and mesoscopic phase-field parameters \cite{Kundin2014}. This limits the applicability of spectral methods due to the lack of periodic boundary conditions in these applications. On the other hand, simulations in complex geometries have been considered, e.g. colloids in confinements motivated by studies of DDFT \cite{Rauscher2011}, crystallization on embedded manifolds \cite{Backofen2010, Backofen2011, Aland2012c, Schmid2014} or particle-stabilized emulsions, where the PFC model is considered at fluid-fluid interfaces \cite{Aland2011, Aland2012}. These applicabilities also limit the use of spectral and finite-difference methods or make them sometimes even impossible. The finite element method provides high flexibility concerning complex geometries and coupling to other partial differential equations, which is the motivation to develop efficient solution methods for the PFC model based on finite element discretizations.

Basic steps of finite element methods include refinement and coarsening of a mesh, error estimation, assembling of a linear system and solving the linear system. Most previous finite element simulations for the PFC model \cite{Backofen2007, Backofen2010, Praetorius2011} have used direct solvers for the last step, which however restrict the system size due to the high memory requirements and only allow computations in 2D. Well-established solution methods for linear systems, such as iterative Krylov-subspace solvers, like CG, MINRES, GMRES, TFQMR, BiCGStab are not directly applicable for the PFC equation or do not converge, respectively converge very slowly, if used without or with standard preconditioners, like Jacobi or ILU preconditioners.

In this paper, we propose a block-preconditioner for the discretized PFC equations and analyze it with respect to convergence properties of a GMRES method. We have organized the paper as follows: In the next section, we formulate the PFC model in terms of a higher order non-linear partial differential equation. Section \ref{seq:discretization} introduces a space- and time-discretization of the model, including the treatment of the non-linearity. In Section \ref{seq:preconditioner}, the preconditioner is introduced and an efficient preconditioning procedure is formulated. The convergence analysis of GMRES is introduced in Section \ref{seq:convergence_analysis} and Section \ref{seq:spectral_analysis} provides an analysis of the preconditioner in terms of a spectral analysis. Finally, in Section \ref{seq:numerical_examples} we examine the preconditioner in numerical examples and demonstrate its efficiency. Conclusion and outlook are provided in Section \ref{seq:conclusion}.

\section{Modelling}\label{seq:modelling}
We consider the original model introduced in \cite{Elder2002}, which is a conserved gradient flow of a Swift-Hohenberg energy and serves as a model system for a regular periodic wave-like order-parameter field that can be interpreted as particle density. The Swift-Hohenberg energy is given here in a simplified form:
\begin{equation}\label{eq:swift-hohenberg-energy} F(\psi) = \int_\Omega \frac{1}{4}\psi^4 + \frac{1}{2}\psi(r+(1+\Delta)^2)\psi\,\text{d}x,\end{equation}
where the order-parameter field $\psi$ describes the deviation from a reference density, the parameter $r$ can be related to the temperature of the system and $\Omega\subset\R^m,\,m=1,2,3$ is the spatial domain. According to the notation in 
\cite{Teeffelen2009} we consider the $H^{-1}$-gradient flow of $F$, the PFC2-model:
\begin{equation}\label{eq:H-1-gradientflow}
\partial_t\psi = \Delta\frac{\delta F[\psi]}{\delta\psi},
\end{equation}
respective a Wasserstein gradient flow \cite{Jordan1997} of $F$, the PFC1-model, as a generalization of \eqref{eq:H-1-gradientflow}:
\begin{equation}\label{eq:wasserstein-gradientflow}
\partial_t\psi = \nabla\cdot \Big(\psi^+\nabla\frac{\delta F[\psi]}{\delta\psi}\Big),
\end{equation}
with a mobility coefficient $\psi^+=\psi-\psi_\text{min}\geq 0$ with the lower bound\footnote{The lower bound $\psi_{\min}=-1.5$ is due to the scaling and shifting of the order-parameter from a physical density with lower bound $0$.} $\psi_\text{min}=-1.5$. By calculus of variations and splitting of higher order derivatives, we can find a set of second order equations, which will be analyzed in this paper:
\begin{equation}\label{eq:pfc-equation}\begin{split}
\mu &= \psi^3+(1+r)\psi+2\Delta\psi+\Delta\omega\qquad\text{ in }\Omega\times[0,T] \\
\partial_t\psi &= \nabla\cdot (\psi^+\nabla\mu) \\
\omega &= \Delta\psi,
\end{split}\end{equation}
for a time interval $[0,T]$ and subject to initial condition $\psi(t=0)=\psi_0$ in $\bar{\Omega}$ and boundary conditions on $\partial\Omega$, e.g. homogeneous Neumann boundary conditions
\[\partial_n\psi=\partial_n\omega=\psi^+\partial_n\mu=0\quad\text{ on }\partial\Omega.\]

\section{Discrete equations}\label{seq:discretization}
To transform the partial differential equation \eqref{eq:pfc-equation} into a system of linear equations, we discretize in space using finite elements and in time using a backward Euler discretization, respective a Rosenbrock discretization scheme. 

Let $\Omega\subset\mathbb{R}^m$ be a regular domain ($m=1,2,3$) with a conforming triangulation $\mathcal{T}_h(\Omega)$ with $h=\max_{T\in\mathcal{T}_h}(h_T)$ a discretization parameter describing the maximal element size in the triangulation. We consider simplicial meshes, i.e. made of line segments in 1D, triangles in 2D and tetrahedra in 3D. Let \[V_h:=\{v\in H^1(\Omega)\;;\; v|_T\in\mathbb{P}_p(T),\;\forall T\in\mathcal{T}_h(\Omega)\}\] be the corresponding finite element space, with $\mathbb{P}_p(T)$ the space of local polynomials of degree  $\leq p$, where we have chosen $p=1,2$ in our simulations. The problem \eqref{eq:pfc-equation} in discrete weak form can be stated as: 

Find $\mu_h,\psi_h,\omega_h\in L_2(0,T;\,V_h)$ with $\psi_h(t=0)=\psi_0\in L_2(\Omega)$, s.t.
\begin{align}
(\mu_h - \psi_h^3 - (1+r)\psi_h,\vartheta_h)_\Omega &+ (2\nabla\psi_h + \nabla\omega_h, \nabla\vartheta_h)_\Omega \notag \\
+\,(\partial_t\psi_h,\vartheta'_h)_\Omega &+ (\psi_h^+\nabla\mu_h, \nabla\vartheta'_h)_\Omega  \label{eq:discrete_pfc} \\
+\, (\omega_h, \vartheta''_h)_\Omega &+ (\nabla\psi_h, \nabla\vartheta''_h)_\Omega = 0 && \forall\vartheta_h, \vartheta'_h, \vartheta''_h\in V_h, \notag
\end{align} 
with $(u,v)_\Omega:=\int_\Omega u\cdot v\,\text{d}x$. 

In the following let $0=t_0<t_1<\ldots<t_N=T$ be a discretization of the time interval $[0,T]$. Let $\tau_k:=t_{k+1}-t_k$ be the timestep width in the $k$-th iteration and $\psi_k\equiv\psi_h(t_k)$, respective $\mu_k\equiv\mu_h(t_k)$ and $\omega_k\equiv\omega_h(t_k)$ the discrete functions at time $t_k$. Applying a semi-implicit Euler discretization to \eqref{eq:discrete_pfc} results in a time and space discrete system of equations:

Let $\psi_0\in L_2(\Omega)$ be given. For $k=0,1,\ldots,N-1$ find $\mu_{k+1},\psi_{k+1},\omega_{k+1}\in V_h$, s.t.
\begin{equation}\label{eq:time_discrete_pfc}\begin{split} 
a^{(e)}\big((\mu_{k+1}, \psi_{k+1},\omega_{k+1}) &, (\vartheta_h, \vartheta'_h, \vartheta''_h)\big) := \\
    (\mu_{k+1} &- (1+r)\psi_{k+1},\vartheta_h)_\Omega + (2\nabla\psi_{k+1} + \nabla\omega_{k+1}, \nabla\vartheta_h)_\Omega \\
	&+ \big(\psi_{k+1},\vartheta'_h\big)_\Omega + (\tau_k \psi_{k}^+\nabla\mu_{k+1}, \nabla\vartheta'_h)_\Omega \\
	&+ (\omega_{k+1}, \vartheta''_h)_\Omega + (\nabla\psi_{k+1}, \nabla\vartheta''_h)_\Omega \\
	=(\psi_{k}^3,\vartheta_h)_\Omega &+ \big(\psi_{k},\vartheta'_h\big)_\Omega =: \big\langle F^{(e)}, (\vartheta_h, \vartheta'_h, \vartheta''_h)\big\rangle\qquad \forall\vartheta_h, \vartheta'_h, \vartheta''_h\in V_h.
\end{split}\end{equation}

Instead of taking $\psi_{k}^3$ explicitly it is pointed out in \cite{Backofen2007}, that a linearization of this non-linear term stabilizes the system and allows for larger timestep widths. Therefore, we replace $(\psi_{k}^3,\vartheta_h)_\Omega$ by $(3\psi_{k}^2\psi_{k+1} - 2\psi_k^3,\vartheta_h)_\Omega$. Thus \eqref{eq:time_discrete_pfc} reads
\begin{equation}\label{eq:time_discrete_pfc_2}\begin{split} 
a\big((\mu_{k+1}, \psi_{k+1},\omega_{k+1}) &, (\vartheta_h, \vartheta'_h, \vartheta''_h)\big) := \\
    (\mu_{k+1} - (1+r)\psi_{k+1} &- 3\psi_{k}^2\psi_{k+1},\vartheta_h)_\Omega + (2\nabla\psi_{k+1} + \nabla\omega_{k+1}, \nabla\vartheta_h)_\Omega \\
	&+ \big(\psi_{k+1},\vartheta'_h\big)_\Omega + (\tau_k \psi_{k}^+\nabla\mu_{k+1}, \nabla\vartheta'_h)_\Omega \\
	&+ (\omega_{k+1}, \vartheta''_h)_\Omega + (\nabla\psi_{k+1}, \nabla\vartheta''_h)_\Omega \\
	=(-2\psi_{k}^3,\vartheta_h)_\Omega &+ \big(\psi_{k},\vartheta'_h\big)_\Omega =: \big\langle F, (\vartheta_h, \vartheta'_h, \vartheta''_h)\big\rangle\qquad \forall\vartheta_h, \vartheta'_h, \vartheta''_h\in V_h,
\end{split}\end{equation}

Let $\{\varphi_i\}$ be a basis of $V_h$, than we can define the system matrix $\mathbf{A}$ and the right-hand side vector $\mathbf{b}$, for the linear system $\mathbf{A}\mathbf{x}=\mathbf{b}$, as
\[\mathbf{A} = \begin{bmatrix}
\mathbf{A}^{00} & \mathbf{A}^{01} & \mathbf{A}^{02} \\
\mathbf{A}^{10} & \mathbf{A}^{11} & \mathbf{A}^{12} \\
\mathbf{A}^{20} & \mathbf{A}^{21} & \mathbf{A}^{22} \\
\end{bmatrix},\quad \mathbf{b} = \begin{pmatrix} \mathbf{b}^{0} \\ \mathbf{b}^{1} \\ \mathbf{b}^{2}\end{pmatrix},\]
with each block defined via
\[[\mathbf{A}^{ij}]_{kl} = a\big(\mathbf{e}_j \varphi_l, \mathbf{e}_i \varphi_k\big),\quad [\mathbf{b}^i]_j = \big\langle F, \mathbf{e}_i\varphi_j\big\rangle,\]
where $\mathbf{e}_i$ is the $i$-th Cartesian unit vector. 

Introducing the short cuts $\mathbf{M}:=\big((\varphi_j,\varphi_i)_\Omega\big)_{ij}$ and $\mathbf{K}:=\big((\nabla\varphi_j,\nabla\varphi_i)_\Omega\big)_{ij}$ for mass- and stiffness-matrix,  $\mathbf{K}_+(\psi):=\big((\psi^+\nabla\varphi_j,\nabla\varphi_i)_\Omega\big)_{ij}$ for the mobility matrix and for the non-linear term the short cut $\mathbf{N}(\psi):=\big((-3\psi^2\varphi_j,\varphi_i)_\Omega\big)_{ij}$, we can write $\mathbf{A}$ as
\begin{equation}\label{eq:systemmatrix}
\mathbf{A} = \begin{bmatrix}
 \mathbf{M} & -(1+r)\mathbf{M} + \mathbf{N}(\psi_k) + 2\mathbf{K} & \mathbf{K} \\
 \tau_k\mathbf{K}_+(\psi_k) & \mathbf{M} & 0 \\
0 & \mathbf{K} & \mathbf{M} \\
\end{bmatrix}.
\end{equation}
We also find that $\mathbf{b}^0=\big((-2\psi_{k}^3,\varphi_j)_\Omega\big)_j$, $\mathbf{b}^1=\mathbf{M}\underline{\psi}_k$ and $\mathbf{b}^2=0$. Using this, we can define a new matrix $\mathbf{B}:=\mathbf{K}\mathbf{M}^{-1}\mathbf{K}$ to decouple the first two equations from the last equation, i.e.
\begin{equation}\label{eq:systemmatrix2}
\mathbf{A}' = \begin{bmatrix}
 \mathbf{M} & -(1+r)\mathbf{M} + \mathbf{N}(\psi_k) + 2\mathbf{K} - \mathbf{B} \\
 \tau_k\mathbf{K}_+(\psi_k) & \mathbf{M}
\end{bmatrix}.
\end{equation}
With $\mathbf{x}=(\underline{\mu}_{k+1},\underline{\psi}_{k+1},\underline{\omega}_{k+1})^\top, \mathbf{x}'=(\underline{\mu}_{k+1},\underline{\psi}_{k+1})^\top, \mathbf{b}'=(\mathbf{b}^0,\mathbf{b}^1)^\top$, where the discrete coefficient vectors correspond to a discretization with the same basis functions as the matrices, i.e.
\[\psi_h = \sum_i\psi_{(i)}\varphi_i\;\text{ with coefficients }\;\underline{\psi} = (\psi_{(i)})_i,\]
and $\underline{\mu}, \underline{\omega}$ in a same manner, we have
\[\text{\eqref{eq:time_discrete_pfc_2}}\; \Leftrightarrow\; \mathbf{A}\mathbf{x}=\mathbf{b}\; \Leftrightarrow \;\mathbf{A}' \mathbf{x}' = \mathbf{b}',\, \mathbf{M}\underline{\omega}_{k+1} = -\mathbf{K}\underline{\psi}_{k+1}.\]
The reduced system can be seen as a discretization of a partial differential equation including the Bi-Laplacian, i.e.
\[\partial_t\psi = \nabla\cdot (\psi^+\nabla\mu),\quad\text{with}\quad \mu = \psi^3+(1+r)\psi+2\Delta\psi+\Delta^2\psi.\]
In the following, we will drop the underscore for the coefficient vectors for ease of reading.

\subsection{Rosenbrock time-discretization}\label{seq:rosenbrock}
To obtain a time discretization with high accuracy and stability with an easy step size control, we replace the discretization \eqref{eq:time_discrete_pfc}, respective \eqref{eq:time_discrete_pfc_2}, by an embedded Rosenbrock time-discretization scheme, see e.g. \cite{Hairer1993II,Lang2010,Rang2005,John2006, John2010}.

Therefore consider the abstract general form of a differential algebraic equation
\[\mathbb{M}\partial_t \mathbf{x} = \mathbb{F}[\mathbf{x}],\]
with a linear (mass-)operator $\mathbb{M}$ and a (non-linear) differential operator $\mathbb{F}$. Using the notation $\mathbb{J}_\mathbb{F}(\mathbf{x})[\mathbf{y}]:=\frac{\text{d}}{\text{d}\epsilon}\mathbb{F}[\mathbf{x}+\epsilon \mathbf{y}]\big|_{\epsilon=0}$ for the G\^{a}teaux derivative of $\mathbb{F}$ at $\mathbf{x}$ in the direction $\mathbf{y}$, one can write down a general Rosenbrock scheme
\begin{gather}\label{eq:rosenbrock}
\frac{1}{\tau_k\gamma} \mathbb{M}\,\mathbf{y}_i^k -\mathbb{J}_\mathbb{F}(\mathbf{x}^k)[\mathbf{y}_i^k] = \mathbb{F}[\mathbf{x}_{i}^k]+\sum_{j=1}^{i-1}\frac{c_{ij}}{\tau_k}\mathbb{M}\,\mathbf{y}_j^k  \\
\qquad\text{ for } i=1,\ldots,s \nonumber\\
\label{eq:rosenbrock-update}\begin{split}
\mathbf{x}_{i}^k &= \mathbf{x}^k+\sum_{j=1}^{i-1} a_{ij}\mathbf{y}_j^k \quad (i^\text{th}\text{ stage solution}) \\
\mathbf{x}^{k+1} &= \mathbf{x}^k + \sum_{j=1}^s m_j \mathbf{y}_j^k,\quad \mathbf{\hat{x}}^{k+1} = \mathbf{x}^k + \sum_{j=1}^s \hat{m}_j \mathbf{y}_j^k,
\end{split}\end{gather}
with coefficients $\gamma,a_{ij},c_{ij},m_i,\hat{m}_i$ and timestep $\tau_k$. The coefficients $m_i$ and $\hat{m}_i$ build up linear-combinations of the intermediate solutions of two different orders. This can be used in order to estimate the timestep error and control the timestep width. Details about stepsize control can be found in \cite{Hairer1993II, Lang2010}. The coefficients used for the PFC equation are based on the \textit{Ros3Pw} scheme \cite{Rang2005} and are listed in Table \ref{tbl:rosenbrock_scheme}. This W-method has 3 internal stages, i.e. $s=3$, and is strongly A-stable. As Rosenbrock-method it is of order 3. It avoids order reduction when applied to semidiscretized parabolic PDEs and is thus applicable to our equations.

\begin{table}[ht]
\begin{center}
\begin{tabular}{|l|l|}
\hline 
$\gamma = 0.78867513459481287$ & $c_{11} = -c_{22} = \gamma$ \\ 
\cline{1-1}
$a_{21} = 2$ & $c_{21} = -2.53589838486225$ \\
$a_{22} = 1.57735026918963$ & $c_{31} = -1.62740473580836$\\
$a_{31} = 0.633974596215561$ & $c_{32} = -0.274519052838329$ \\
$a_{33} = 0.5$ & $c_{33} = -0.0528312163512967$ \\
 \hline
$m_1 = 1.63397459621556$ & $\hat{m}_1 = 1.99444650053487$ \\ 
$m_2 = 0.294228634059948$ & $\hat{m}_2 = 0.654700538379252$ \\ 
$m_3 = 1.07179676972449$ & $\hat{m}_3 = m_3$ \\ 
\hline 
\end{tabular} 
\end{center}
\caption{A set of coefficients for the \textit{Ros3Pw} Rosenbrock scheme translated into the modified form of the Rosenbrock method used in \eqref{eq:rosenbrock}. All coefficients not given explicitly are set to zero.}\label{tbl:rosenbrock_scheme}
\end{table}

In case of the PFC system \eqref{eq:pfc-equation} we have $\mathbf{x}=(\mu,\psi,\omega)^\top$ and $\mathbb{M}=\operatorname{diag}(0,1,0)$. The functional $\mathbb{F}$ applied to $\mathbf{x}$ is given by 
\begin{equation}\label{eq:non-linear-operator}
\mathbb{F}[\mathbf{x}] = \underbrace{\begin{bmatrix}
-\mu + (1+r)\psi + 2\Delta\psi + \Delta\omega \\
0 \\
-\omega + \Delta\psi
\end{bmatrix}}_{\mathbb{F}_\text{Lin}[\mathbf{x}]} + \begin{bmatrix}
\psi^3 \\ \nabla\cdot (\psi^+\nabla\mu) \\   0
\end{bmatrix}.\end{equation}
For the Jacobian of $\mathbb{F}$ in the direction $\mathbf{y}=(d\mu, d\psi, d\omega)^\top$ we find
\begin{align*}
\mathbb{J}_\mathbb{F}(\mathbf{x})[\mathbf{y}] &= \mathbb{F}_\text{Lin}[\mathbf{y}] + \begin{bmatrix}
3\psi^2d\psi \\ \nabla\cdot (\psi^+\nabla d\mu) + \nabla\cdot (d\psi\partial_\psi(\psi^+)\nabla\mu) \\ 0
\end{bmatrix} \\
& (\text{assuming }\psi^+=\psi-\psi_{\min}) \\
&= \mathbb{F}_\text{Lin}[\mathbf{y}] + \begin{bmatrix}
3\psi^2d\psi \\ \nabla\cdot (\psi^+\nabla d\mu) + \nabla\cdot (d\psi\nabla\mu) \\ 0
\end{bmatrix}.
\end{align*}
By multiplication with test functions $\boldsymbol{\vartheta}=(\vartheta,\vartheta',\vartheta'')^\top$ and integration over $\Omega$, we can derive a weak form of equation \eqref{eq:rosenbrock}: 

For $i=1,\ldots,s$ find $\mathbf{y}_i^k\in \big(L_2(0,T;\,V_h)\big)^3$, s.t.
\begin{equation}\label{eq:rosenbrock_weak}
\frac{1}{\tau_k\gamma} (\mathbb{M}\,\mathbf{y}_i^k,\boldsymbol{\vartheta})_\Omega -\mathbb{J}_\mathbb{F}(\mathbf{x}^k)[\mathbf{y}_i^k,\boldsymbol{\vartheta}] = \mathbb{F}(\mathbf{x}_{i}^k)[\boldsymbol{\vartheta}]+\sum_{j=1}^{i-1}\frac{c_{ij}}{\tau_k}(\mathbb{M}\,\mathbf{y}_j^k, \boldsymbol{\vartheta})_\Omega\qquad\forall\boldsymbol{\vartheta}\in(V_h)^3,
\end{equation}
with the linear form $\mathbb{F}(\cdot)[\cdot]$:
\[\begin{split}
\mathbb{F}(\mathbf{x})[\boldsymbol{\vartheta}] &= \Big[(-\mu + (1+r)\psi,\vartheta)_\Omega - (2\nabla\psi+\nabla\omega,\nabla\vartheta)_\Omega - (\omega,\vartheta'')_\Omega - (\nabla\psi,\nabla\vartheta'')_\Omega\Big]\\
	&\quad  + (\psi^3,\vartheta)_\Omega -(\psi^+\nabla\mu,\nabla\vartheta')_\Omega \\
	&=: \mathbb{F}_\text{Lin}(\mathbf{x})[\boldsymbol{\vartheta}]  + (\psi^3,\vartheta)_\Omega -(\psi^+\nabla\mu,\nabla\vartheta')_\Omega
\end{split}\]
and the bi-linear form $\mathbb{J}_\mathbb{F}(\cdot)[\cdot, \cdot]$:
\[\mathbb{J}_\mathbb{F}(\mathbf{x})[\mathbf{y},\boldsymbol{\vartheta}] = \mathbb{F}_\text{Lin}(\mathbf{y})[\boldsymbol{\vartheta}] + (3\psi^2d\psi,\vartheta)_\Omega - (\psi^+\nabla d\mu + d\psi\nabla\mu,\nabla\vartheta')_\Omega.\]

Using the definitions of the elementary matrices $\mathbf{M}, \mathbf{K}, \mathbf{K}_+$ and $\mathbf{N}$, as above and introducing $\mathbf{F}(\mu):=\big((\varphi_j\nabla\mu,\nabla\varphi_i)_\Omega\big)_{ij}$, we can write the Rosenbrock discretization in matrix form for the $i$-th stage iteration:
\begin{equation}\label{eq:rosenbrock_matrix}\underbrace{\begin{bmatrix}
 \mathbf{M} & -(1+r)\mathbf{M} + 2\mathbf{K} + \mathbf{N}(\psi_k) & \mathbf{K} \\
 \tau_k\mathbf{K}_+(\psi_k) & \frac{1}{\gamma}\mathbf{M} + \tau_k\mathbf{F}(\mu_k) & 0 \\
0 & \mathbf{K} & \mathbf{M} \\
\end{bmatrix}}_{\mathbf{A}^R}\mathbf{y}_i^k = \mathbf{b}_i^R,\end{equation}
with $\mathbf{b}_i^R$ the assembling of the right-hand side of \eqref{eq:rosenbrock_weak}, with a factor $\tau_k$ multiplied to the second component. The system matrix $\mathbf{A}^R$ in each stage of one Rosenbrock time iteration is very similar to the matrix derived for the simple backward Euler discretization in \eqref{eq:systemmatrix}, up to a factor $\frac{1}{\gamma}$ in front of a mass matrix and the derivative of the mobility term $\mathbf{F}$. The latter can be simplified in case of the PFC2 model \eqref{eq:H-1-gradientflow}, where $\mathbf{F}=0$ and $\mathbf{K}_+=\mathbf{K}$.

\section{Precondition the linear systems}\label{seq:preconditioner}
To solve the linear system $\mathbf{A}\mathbf{x}=\mathbf{b}$, respective $\mathbf{A}^R\mathbf{y}=\mathbf{b}^R$, linear solvers must be applied. As direct solvers, like UMFPACK \cite{Davis2004}, MUMPS \cite{MUMPS} or SuplerLU\_DIST \cite{SuperLU} suffer from fast increase of memory requirements and bad scaling properties for massively parallel problems,  iterative solution methods are required. The system matrix $\mathbf{A}$, respective $\mathbf{A}^R$, is non-symmetric, non-positive definite and non-normal, which restricts the choice of applicable solvers. We here use a GMRES algorithm \cite{Saad1986}, respectively the flexible variant FGMRES \cite{Saad1993}, to allow for preconditioners with (non-linear) iterative inner solvers, like a CG method. 

Instead of solving the linear system $\mathbf{A} \mathbf{x}=\mathbf{b}$, we consider the modified system $\mathbf{A} \mathbf{P}^{-1} (\mathbf{P} \mathbf{x}) = \mathbf{b}$, i.e. a right preconditioning of the matrix $\mathbf{A}$. A natural requirement for the preconditioner $\mathbf{P}$ is, that it should be simple and fast to solve $\mathbf{P}^{-1} \mathbf{v}$ for arbitrary vectors $\mathbf{v}$, since it is applied to the Krylov basisvectors in each iteration of the (F)GMRES method.

We propose a block-preconditioner $\mathbf{P}$ for the $2\times 2$ upper left block matrix $\mathbf{A}'$ of $\mathbf{A}$ based on an approach similar to a preconditioner developed for the Cahn-Hilliard equation \cite{Boyanova2012}. Therefore, we first simplify the matrix $\mathbf{A}'$, respective the corresponding reduced system ${\mathbf{A}^R}'$ of $\mathbf{A}^R$, by considering a fixed timestep $\tau_k=\tau$ and using a constant mobility approximation, i.e. $\mathbf{K}_+\approx M_0\mathbf{K}$, with $M_0=\langle \psi^+\rangle$ the mean of the mobility coefficient $\psi^+$, and $\mathbf{F}=0$. For simplicity, we develop the preconditioner for the case $M_0=1$ and $\gamma=1$ only. For small timestep widths $\tau$ the semi-implicit Euler time-discretization \eqref{eq:time_discrete_pfc} is a good approximation of \eqref{eq:rosenbrock_matrix}, so we neglect the non-linear term $\mathbf{N}(\psi)$. What remains is the reduced system
\[\mathbf{A}'':=\begin{bmatrix}
 \mathbf{M} & -(1+r)\mathbf{M} + 2\mathbf{K} - \mathbf{B} \\
 \tau\mathbf{K} & \mathbf{M}
\end{bmatrix}.\]
By adding a small perturbation to the diagonal of $\mathbf{A}''$, we can find a matrix having an explicit triangular block-factorization. This matrix we propose as a preconditioner for the original matrix $\mathbf{A}'$:
\begin{equation}\label{eq:pfc-preconditioner}\begin{split}
\!\!\!\mathbf{P}&:=\left[\begin{array}{cc}
\mathbf{M} & 2\mathbf{K} - \mathbf{B} \\
\tau\mathbf{K} & \mathbf{M} - \delta\mathbf{K}+\delta\mathbf{B}
\end{array}\right]
= \left[\begin{array}{cc}
\mathbf{M} & 0\\
\tau\mathbf{K} & \mathbf{M} + \delta\mathbf{K}
\end{array}\right]\left[\begin{array}{cc}
\mathbf{I} & \mathbf{M}^{-1}(2\mathbf{K} - \mathbf{B}) \\
0 & \mathbf{M}^{-1}(\mathbf{M} - 2\delta\mathbf{K}+\delta\mathbf{B})
\end{array}\right]
\end{split}\end{equation}
with $\delta:=\sqrt{\tau}$. In each (F)GMRES iteration, the preconditioner is applied to a vector $(\mathbf{b}_0, \mathbf{b}_1)^\top$, that means solving the linear system $\mathbf{P}\mathbf{x}=\mathbf{b}$, in four steps:
\begin{align*}
(1)\; & \mathbf{M} \mathbf{y}_0 = \mathbf{b}_0 & (2)\; & (\mathbf{M}+\delta\mathbf{K}) \mathbf{y}_1 = \mathbf{b}_1 - \tau\mathbf{K} \mathbf{y}_1 \\
(3)\; & (\mathbf{M} - 2\delta\mathbf{K}+\delta\mathbf{B}) \mathbf{x}_1=\mathbf{M} \mathbf{y}_1 & (4)\; & \mathbf{x}_0= \mathbf{y}_0+\frac{1}{\delta}(\mathbf{y}_1 - \mathbf{x}_1).
\end{align*}
Since the overall system matrix $\mathbf{A}$ has a third component, that was removed for the construction of the preconditioner, the third component $\mathbf{b}_2$ of the vector has to be preconditioned as well. This can be performed by solving: (5) $\mathbf{M} \mathbf{x}_2=\mathbf{b}_2 - \mathbf{K}\mathbf{x}_1$.

In step (3) we have to solve
\begin{equation}\label{eq:schur_complement_step_3}
\mathbf{S} \mathbf{x}_1 := (\mathbf{M} - 2\delta\mathbf{K}+\delta\mathbf{K}\mathbf{M}^{-1}\mathbf{K}) \mathbf{x}_1=\mathbf{M} \mathbf{y}_1,
\end{equation}
which requires special care, as forming the matrix $\mathbf{S}$ explicitly is no option, as the inverse of the mass matrix $\mathbf{M}$ is dense and thus the matrix $\mathbf{S}$ as well. In the following subsections we give two approximations to solve this problem. 

\subsection{Diagonal approximation of the mass matrix}
Approximating the mass matrix by a diagonal matrix leads to a sparse approximation of $\mathbf{S}$. Using the ansatz $\mathbf{M}^{-1}\approx\mathbf{\operatorname{diag}(M)}^{-1}=:\mathbf{M}_D^{-1}$ the matrix $\mathbf{S}$ can be approximated by
\[\mathbf{S}_D := (\mathbf{M} - 2\delta\mathbf{K}+\delta\mathbf{K}\mathbf{M}_D^{-1}\mathbf{K}).\]
By estimating the eigenvalues of the generalized eigenvalue problem $\lambda \mathbf{S}_D \mathbf{x} = \mathbf{S} \mathbf{x}$ we show, similar as in \cite{Boyanova2012b}, that the proposed matrix is a good approximation. 
\begin{lem}\label{lem:diagonal_mass_approx}
The eigenvalues $\lambda$ of the generalized eigenvalue problem $\lambda \mathbf{S}_D \mathbf{x} = \mathbf{S} \mathbf{x}$ are bounded by bounds of the eigenvalues $\mu$ of the generalized eigenvalue problem $\mu \mathbf{M}_D \mathbf{y} = \mathbf{M} \mathbf{y}$ for mass-matrix and diagonal approximation of the mass-matrix.
\end{lem}
\begin{proof} We follow the argumentation of \cite[Section 3.2]{Boyanova2012b}.

Using the matrices $\widehat{\mathbf{D}}:=\mathbf{M}^{\frac{1}{2}}\mathbf{M}_D^{-1}\mathbf{M}^{\frac{1}{2}}$ and $\widehat{\mathbf{K}}:=\mathbf{M}^{-\frac{1}{2}}\mathbf{K}\mathbf{M}^{-\frac{1}{2}}$ we can reformulate the eigenvalue problem $\lambda \mathbf{S}_D \mathbf{x} = \mathbf{S} \mathbf{x}$ as
\begin{equation}\label{eq:eigenvalue_problem_diagonal}
\lambda\mathbf{M}^{\frac{1}{2}}(\mathbf{I} - 2\delta\widehat{\mathbf{K}}+\delta\widehat{\mathbf{K}}\widehat{\mathbf{D}}\widehat{\mathbf{K}})\mathbf{M}^{\frac{1}{2}} \mathbf{x} = \mathbf{M}^{\frac{1}{2}}(\mathbf{I}-2\delta\widehat{\mathbf{K}}+\delta\widehat{\mathbf{K}}\widehat{\mathbf{K}})\mathbf{M}^{\frac{1}{2}} \mathbf{x}.
\end{equation}
Multiplying from the left with $\mathbf{x}^\top$, dividing by $\|\mathbf{M}^{\frac{1}{2}} \mathbf{x}\|^2$ and defining the normalized vector $\mathbf{y}:=\mathbf{M}^{\frac{1}{2}}\mathbf{x} / \|\mathbf{M}^{\frac{1}{2}}\mathbf{x}\|$ results in a scalar equation for $\lambda$:
\[\lambda(1-2\delta k + \delta k^2 d) = 1-2\delta k + \delta k^2,\]
with the Rayleigh quotients $k=\mathbf{y}^\top\widehat{\mathbf{K}}\mathbf{y} / (\mathbf{y}^\top \mathbf{y})$ and $d=\mathbf{y}^\top\widehat{\mathbf{D}}\mathbf{y} / (\mathbf{y}^\top \mathbf{y})$. Assuming that $(1-2\delta k + \delta k^2 d) \neq 0$ we arrive at
\[\lambda = \frac{1-2\delta k + \delta k^2}{1-2\delta k + \delta k^2 d},\]
where the difference in the highest order terms of the rational function is the factor $d$. From the definition of $d$ and $\widehat{\mathbf{D}}$, bounds are given by the bounds of the eigenvalues of $\mu\mathbf{M}_D v = \mathbf{M}v$. 
\end{proof}

In \cite{Wathen1987} concrete values are provided for linear and quadratic Lagrangian finite elements on triangles and linear Lagrangian elements on tetrahedra. For the latter, the bound $d\in[0.3924, 2.5]$ translates directly to the bound for $\lambda$, i.e. $\lambda\in[0.3924, 2.5]$, and thus $\mathbf{S}_D$ provides a reasonable approximation of $\mathbf{S}$.

\begin{rem}
Other diagonal approximations based on lumped mass matrices could also be used, which however would lead to different eigenvalue bounds.
\end{rem}

\subsection{Relation to a Cahn-Hilliard system}
An alternative to the diagonal approximation can be achieved by using the similarity of step (3) in the preconditioning with the discretization of a Cahn-Hilliard equation \cite{Cahn1958,Boyanova2012}. This equation can be written using higher order derivatives:
\[\partial_t c = \Delta(c^3)-\Delta c - \eta\Delta^2 c\]
with $\eta$ a parameter related to the interface thickness. For an Euler discretization in time with timestep width $\tau'$ and finite element discretization in space as above, we find the discrete equation
\[\big(\mathbf{M} - \tau'\mathbf{K}+\tau'\eta\mathbf{B}-\tau'\mathbf{N}'(c_k)\big)c_{k+1} = \mathbf{M}c_k.\]
Setting $\eta:=\frac{1}{2}$ and $\tau':=2\delta$, and neglecting the Jacobian operator $\mathbf{N}'$ we recover equation \eqref{eq:schur_complement_step_3}. A preconditioner for the Cahn-Hilliard equation, see \cite{Boyanova2012,Boyanova2012b,Axelsson2013} thus might help to solve the equation in step (3), which we rewrite as a block system
\begin{equation}\label{eq:cahn_hilliard_3c}\begin{bmatrix}
\mathbf{M} & \mathbf{M}-\eta\mathbf{K} \\
\tau'\mathbf{K} & \mathbf{M}
\end{bmatrix}\begin{pmatrix}\ast \\ x_1\end{pmatrix} = \begin{pmatrix}0 \\ \mathbf{M} y_1\end{pmatrix}\end{equation}
with Schur complement $\mathbf{S}$. Using the proposed inner preconditioner $\widehat{\mathbf{A}}_0$ of \cite[p.13]{Boyanova2012}:
\[\widehat{\mathbf{A}}_0:=\begin{bmatrix}
\mathbf{M} & -\eta\mathbf{K} \\
\tau'\mathbf{K} & \mathbf{M} + 2\sqrt{\tau' \eta}\mathbf{K}
\end{bmatrix},\]
with Schur complement
$\mathbf{S}_\text{CH} := \mathbf{M} + 2\sqrt{\tau' \eta}\mathbf{K} + \tau' \eta\mathbf{K}\mathbf{M}^{-1}\mathbf{K}$
as a direct approximation of equation \eqref{eq:cahn_hilliard_3c}, respective \eqref{eq:schur_complement_step_3}, i.e.
\begin{equation}\label{eq:cahn_hilliard_precon_2}
\mathbf{S}_\text{CH} \mathbf{x}_1 =(\mathbf{M} + 2\sqrt{\delta}\mathbf{K} + \delta\mathbf{K}\mathbf{M}^{-1}\mathbf{K}) \mathbf{x}_1=\mathbf{M} \mathbf{y}_1
\end{equation}
we arrive at a simple two step procedure for step (3):
\begin{equation*}
(3.1) \; (\mathbf{M}+\sqrt{\delta}\mathbf{K}) \mathbf{z} = \mathbf{M} \mathbf{y}_1 \qquad
(3.2) \;(\mathbf{M}+\sqrt{\delta}\mathbf{K}) \mathbf{x}_1 = \mathbf{M} \mathbf{z}.
\end{equation*}
\begin{lem}\label{lem:ch_approx}
The eigenvalues $\lambda$ of the generalized eigenvalue problem $\lambda\mathbf{S}_\text{CH} \mathbf{x} = \mathbf{S} \mathbf{x}$ satisfy $\lambda\in[(1-\sqrt{\delta})/2,1]$.
\end{lem}
\begin{proof}
We follow the proof of \cite[Theorem 4]{Pearson2012} and denote by $\lambda$ the eigenvalue of $\mathbf{S}_\text{CH}^{-1}\mathbf{S}$ with the corresponding eigenvector $\mathbf{x}$. We have $\mathbf{M}$ symmetric and positive definite and hence $\mathbf{I}+\sqrt{\delta}\mathbf{M}^{-1}\mathbf{K}$ positive definite and thus invertible.
\begin{align*}
&&\mathbf{S}_\text{CH}^{-1}\mathbf{S} \mathbf{x} &= \lambda \mathbf{x}\\
\Rightarrow&& (\mathbf{M} + 2\sqrt{\delta}\mathbf{K}+\delta\mathbf{K}\mathbf{M}^{-1}\mathbf{K})^{-1}(\mathbf{M} - 2\delta\mathbf{K}+\delta\mathbf{K}\mathbf{M}^{-1}\mathbf{K})\mathbf{x} &= \lambda \mathbf{x} \\
\Rightarrow&& (\mathbf{I} + \sqrt{\delta}\mathbf{M}^{-1}\mathbf{K})^{-2}\big(\mathbf{I} - 2\delta\mathbf{M}^{-1}\mathbf{K}+ \delta(\mathbf{M}^{-1}\mathbf{K})^2\big)\mathbf{x} &= \lambda\mathbf{x}.
\end{align*}
Thus, for each eigenvalue $\mu$ of $\mathbf{M}^{-1}\mathbf{K}$ we have $\mu\in\mathbb{R}_{\geq 0}$ and \[\lambda(\mu):=(\mu^2 + 2 \delta\mu + \delta)(\mu+\sqrt{\delta})^{-2}\] an eigenvalue of $\mathbf{S}_\text{CH}^{-1}\mathbf{S}$ and since $\mathbf{M}^{-1}\mathbf{K}$ is similar to $\mathbf{M}^{1/2}\mathbf{M}^{-1}\mathbf{K}\mathbf{M}^{-1/2}$ that is symmetric, all eigenvalues are determined.

With algebraic arguments and $\sqrt{\delta}>0$ we find
\[\lambda(\mu) \leq \frac{\mu^2 + (\sqrt{\delta})^2}{(\mu+\sqrt{\delta})^2}\leq 1\]
and $\nabla\lambda = 0 $ for $\mu,\delta\searrow 0$. This leads to the lower bound $\frac{1-\sqrt{\delta}}{2} \leq \lambda(\mu)$.
\end{proof}

With this Lemma we can argue that $\mathbf{S}_\text{CH}$ provides a good approximation of $\mathbf{S}$ at least for small timestep widths $\tau = \delta^2 < 1$.

We can write the matrix $\mathbf{P}=\mathbf{P}(\mathbf{S})$ in terms of the Schur complement matrix $\mathbf{S}$:
\begin{equation}\label{eq:preconditioner_S}\mathbf{P}(\mathbf{S})=\left[\begin{array}{cc}
\mathbf{M} & \delta^{-1}(\mathbf{M} - \mathbf{S}) \\
\tau\mathbf{K} & \delta\mathbf{K}+\mathbf{S}
\end{array}\right].\end{equation}
Inserting $\mathbf{S}_\text{CH}$ instead of $\mathbf{S}$ gives the precondition-matrix for the Cahn-Hilliard approximation
\[\mathbf{P}_\text{CH} := \mathbf{P}(\mathbf{S}_\text{CH}) = \left[\begin{array}{cc}
\mathbf{M} & -2\sqrt{\delta}\mathbf{K} - \mathbf{B} \\
\tau\mathbf{K} & \mathbf{M} + (\delta+2\sqrt{\delta})\mathbf{K}+\delta\mathbf{B}
\end{array}\right].\]

\section{Convergence analysis of the Krylov-subspace method}\label{seq:convergence_analysis}
To analyze the proposed preconditioners for the GMRES algorithm, we have a look at the norm of the residuals $\mathbf{r}_k(\mathbf{A})=\mathbf{b}-\mathbf{A}\mathbf{x}_k$ of the approximate solution $\mathbf{x}_k$ obtained in the $k$-th step of the GMRES algorithm. In our studies, we are interested in estimates of the residual norm of the form
\begin{equation}\label{eq:residual_estimate}
\frac{\|\mathbf{r}_k\|_2}{\|\mathbf{r}_0\|_2}=\min_{p\in\Pi_k}\frac{\|p(\mathbf{A})\mathbf{r}_0\|_2}{\|\mathbf{r}_0\|_2} \leq \min_{p\in\Pi_k}\|p(\mathbf{A})\|_2
\end{equation}
with $\Pi_k:=\{p\in\mathbb{P}_k\,:\,p(0)=1\}$ and $\mathbf{r}_0$ the initial residual. The right-hand side corresponds to an ideal-GMRES bound that excludes the influence of the initial residual. In order to get an idea of the convergence behavior, we have to estimate / approximate the right-hand side term by values that are attainable by analysis of $\mathbf{A}$. Replacing $\mathbf{A}$ by $\mathbf{A}\mathbf{P}^{-1}$ we hope to get an improvement in the residuals.

A lower bound for the right-hand side of \eqref{eq:residual_estimate} can be found by using the spectral mapping theorem $p(\sigma(\mathbf{A}))=\sigma(p(\mathbf{A}))$, as
\begin{equation}\label{eq:lower_bound_S}\min_{p\in\Pi_k}\max_{\lambda\in\sigma(\mathbf{A})} |p(\lambda)| \leq \min_{p\in\Pi_k}\|p(\mathbf{A})\|_2, \end{equation}
see \cite{Trefethen2005, Driscoll1998} and an upper bound can be stated by finding a set $\mathcal{S}(\mathbf{A})\subset\mathbb{C}$ associated with $\mathbf{A}$, so that
\begin{equation}\label{eq:upper_bound_S}\min_{p\in\Pi_k}\|p(\mathbf{A})\|_2\leq C \min_{p\in\Pi_k}\max_{\lambda\in\mathcal{S}(\mathbf{A})} |p(\lambda)|, \end{equation}
where $C$ is a constant that depends on the condition number of the eigenvector matrix, the $\epsilon$-pseudospectra of $\mathbf{A}$ respective on the fields of values of $\mathbf{A}$. 

Both estimates contain the min-max value of $p(\lambda)$. In \cite{Ransford1995, Driscoll1998} it is shown, that the limit 
\[\lim_{k\rightarrow\infty}\left[\min_{p\in\Pi_k}\max_{\lambda\in\mathcal{S}} |p(\lambda)|\right]^{1/k} =: \rho_\mathcal{S}\]
exists, where $\rho_\mathcal{S}$ is called the estimated asymptotic convergence factor related to the set $\mathcal{S}$. Thus, for large $k$ we expect a behavior for the right-hand side of \eqref{eq:residual_estimate} like
\[\rho_{\sigma(\mathbf{A})}^k \lesssim \min_{p\in\Pi_k}\|p(\mathbf{A})\|_2\lesssim C \rho_{\mathcal{S}(\mathbf{A})}^k. \]
The tilde indicates that this estimate only holds in the limit $k\rightarrow\infty$.

In the next two sections, we will summarize known results on how to obtain the asymptotic convergence factors $\rho_\mathcal{S}$ and the constant $C$ in the approximation of the relative residual bound.

\subsection{The convergence prefactor}\label{sec:convergence_prefactor}
The constant $C$ plays an important role in the case of non-normal matrices, as pointed out by \cite{Embree1999,Trefethen2005}, and can dominate the convergence in the first iterations. It is shown in Section \ref{seq:spectral_analysis} that the linear part of the operator matrix related to $\mathbf{A}$ is non-normal and also the preconditioned operator related to $\mathbf{Q}:=\mathbf{A}\mathbf{P}^{-1}$ is non-normal. Thus, we have to take a look at this constant.

%

An estimate of the convergence constant, applicable for general non-normal matrices, is related to the $\epsilon$-pseudospectrum $\sigma_\epsilon(\mathbf{A})$ of the matrix $\mathbf{A}$. This can be defined by the spectrum of a perturbed matrix \cite{Trefethen2005, Embree1999}
\[\sigma_\epsilon(\mathbf{A}):=\left\{ z\in\mathbb{C}\,\big|\, z\in \sigma(\mathbf{A}+ \mathbf{E}),\;\|\mathbf{E}\|_2\leq\epsilon\right\}.\]

Let $\Gamma_\epsilon := \partial\sigma_\epsilon$ be the boundary of $\sigma_\epsilon$, respective an union of Jordan curves approximating the boundary, then

\begin{equation}\label{eq:upper_bound_min_max}
\min_{p\in\Pi_k}\|p(\mathbf{A})\|_2 \leq \frac{|\Gamma_\epsilon|}{2\pi\epsilon} \min_{p\in\Pi_k}\max_{\lambda\in\Gamma_\epsilon} |p(\lambda)| \leq \frac{|\Gamma_\epsilon|}{2\pi\epsilon}\min_{p\in\Pi_k}\max_{\lambda\in\sigma_\epsilon(\mathbf{A})} |p(\lambda)|,
\end{equation}
and thus $C\equiv \frac{|\Gamma_\epsilon|}{2\pi\epsilon}$, with $|\Gamma_\epsilon|$ the length of the curve $\Gamma_\epsilon$ \cite{Trefethen2005}. This estimate is approximated, using the asymptotic convergence factor for large $k$, by 
\begin{equation}\label{eq:upper_bound_rho}
\min_{p\in\Pi_k}\|p(\mathbf{A})\|_2 \lesssim \frac{|\Gamma_\epsilon|}{2\pi\epsilon} \rho_{\Gamma_\epsilon}^k \leq \frac{|\Gamma_\epsilon|}{2\pi\epsilon} \rho_{\sigma_\epsilon(\mathbf{A})}^k.
\end{equation}

This constant gives a first insight into the convergence behavior of the GMRES method for the PFC matrix $\mathbf{A}$ respective the preconditioned matrix $\mathbf{Q}$.

\subsection{The asymptotic convergence factor}
The asymptotic convergence factor $\rho_\mathcal{S}$, where $\mathcal{S}$ is a set in the complex plane, e.g. $\mathcal{S}=\sigma(\mathbf{A})$, or $\mathcal{S}=\sigma_\epsilon(\mathbf{A})$, can be estimated by means of potential theory \cite{Driscoll1998,Kuijlaars2005}. Therefore, we have to construct a conformal mapping $\Phi:\mathbb{C}\rightarrow\mathbb{C}$ of the exterior of $\mathcal{S}$ to the exterior of the unit disk with $\Phi(\infty)=\infty$. We assume that $\mathcal{S}$ is connected. Otherwise, we will take a slightly larger connected set. Having $\mathcal{S}\subset\mathbb{C}\setminus\{0\}$ the convergence factor is then given by
\begin{equation}\label{eq:convergence_factor_S}
\rho_\mathcal{S}=\frac{1}{|\Phi(0)|}.
\end{equation}
Let $\mathcal{S}=[\alpha,\beta]$ be a real interval with $0<\alpha<\beta$ and $\kappa:=\frac{\beta}{\alpha}$, then a conformal mapping from the exterior of the interval to the exterior of the unit circle is given by
\begin{equation}\label{eq:Phi_interval}\Phi(z)=\frac{2z-\kappa-1-2\sqrt{z^2-(\kappa+1)z+\kappa}}{\kappa-1},\end{equation}
see \cite{Driscoll1998}\footnote{In \cite{Driscoll1998} the sign of the square-root is wrong and thus, the exterior of the interval is mapped to the interior of the unit circle. In formula \eqref{eq:Phi_interval} this has been corrected.}, and gives the asymptotic convergence factor
\begin{equation}\label{eq:estimate_interval}
\rho_{[\alpha,\beta]}=\frac{\sqrt{\kappa}-1}{\sqrt{\kappa}+1},
\end{equation}
that is a well known convergence bound for the CG method for symmetric positive definite matrices with $\kappa=\frac{\lambda_{\max}}{\lambda_{\min}}$ the spectral condition number of the matrix $\mathbf{A}$. In case of non-normal matrices, the value $\kappa$ is not necessarily connected to the matrix condition number.


In the next Section, we will apply the given estimates for the asymptotic convergence factor and for the constant $C$ to the Fourier transform of the operators that define the PFC equation, in order to get an estimate of the behavior of the GMRES solver.

\section{Spectral analysis of the preconditioner}\label{seq:spectral_analysis}
We analyze the properties and quality of the proposed preconditioner by means of a Fourier analysis. We therefore consider an unbounded, respective periodic domain $\Omega$ and introduce continuous operators $\mathbb{A}, \mathbb{A}_0$ and $\mathbb{P}$ associated with the linear part $\mathbb{F}_\text{Lin}$ of \eqref{eq:non-linear-operator}, of the linear part of the 6th order non-splitted version of \eqref{eq:H-1-gradientflow} and the preconditioner, for $\psi^+\equiv 1$:
\begin{align*}
\mathbb{A}[\mathbf{x}] &:=\begin{bmatrix}
\mu - (1+r)\psi - 2\Delta\psi - \Delta\omega \\
-\tau\Delta\mu + \psi \\
\omega - \Delta\psi
\end{bmatrix}, \\
\mathbb{A}_0[\psi] &:= \psi-\tau\Delta\big((1+r)\psi + 2\Delta\psi + \Delta^2\psi\big),
\end{align*}
with $\mathbf{x}=(\mu,\psi,\omega)$. The operator, that represents the preconditioner reads
\[\mathbb{P}[\mathbf{x}] := \begin{bmatrix}
\mu - 2\Delta\psi - \Delta^2\psi \\
-\tau\Delta\mu + \psi -\delta\Delta\psi +\delta\Delta^2\psi \\
\omega - \Delta\psi
\end{bmatrix}.\]
Using the representation of $\mathbf{P}(\mathbf{S})$ in \eqref{eq:preconditioner_S}, we can also formulate the operator that determines the Cahn-Hilliard approximation of $\mathbf{P}$ by inserting $\mathbf{S}_\text{CH}$:
\[\mathbb{P}_\text{CH}[\mathbf{x}] := \begin{bmatrix}
\mu + 2\sqrt{\delta}\Delta\psi - \Delta^2\psi \\
-\tau\Delta\mu + \psi +(\delta+2\sqrt{\delta})\Delta\psi +\delta\Delta^2\psi \\
\omega - \Delta\psi
\end{bmatrix}.\]

We denote by $\mathbf{k}=(k_1, k_2, k_3)$ the wave vector with $\mathbf{k}^2=k_1^2+k_2^2+k_3^2$. The Fourier transform of a function $u=u(\mathbf{r})$ will be denoted by $\widehat{u}=\widehat{u}(\mathbf{k})$ and is defined as
\[\mathcal{F}:u(\mathbf{r})\mapsto\widehat{u}(\mathbf{k})=\int_{\mathbb{R}^3} e^{-\text{i}(\mathbf{k}\cdot\mathbf{r})}u(\mathbf{r})\,\text{d}\mathbf{r}.\]
Using the inverse Fourier transform, the operators $\mathbb{A}, \mathbb{A}_0, \mathbb{P}$ and $\mathbb{P}_\text{CH}$ applied to $\mathbf{x}$, respective $\psi$, can be expressed as $\mathbb{A}[\mathbf{x}] = \mathcal{F}^{-1}(\mathcal{A}\widehat{\mathbf{x}}),\; \mathbb{A}_0[\psi] = \mathcal{F}^{-1}(\mathcal{A}_0\widehat{\psi}),\;\mathbb{P}[\mathbf{x}] = \mathcal{F}^{-1}(\mathcal{P}\widehat{\mathbf{x}})\text{ and }\mathbb{P}_\text{CH}[\mathbf{x}] = \mathcal{F}^{-1}(\mathcal{P}_\text{CH}\widehat{\mathbf{x}})$, with $\widehat{\mathbf{x}}=(\widehat{\mu},\widehat{\psi},\widehat{\omega})$ and $\mathcal{A}_0, \mathcal{A}, \mathcal{P}$ and $\mathcal{P}_\text{CH}$ the symbols of $\mathbb{A}_0, \mathbb{A}, \mathbb{P}$ and $\mathbb{P}_\text{CH}$, respectively. These symbols can be written in terms of the wave vector $\mathbf{k}$:
\begin{align}
\mathcal{A}\widehat{\mathbf{x}} &= \begin{bmatrix}
1 & -(1+r) +2\mathbf{k}^2 & \mathbf{k}^2 \\
\tau \mathbf{k}^2 & 1 & 0 \\
0 & \mathbf{k}^2 & 1
\end{bmatrix}\begin{pmatrix}
\widehat{\mu} \\ \widehat{\psi} \\ \widehat{\omega}
\end{pmatrix},\\
\mathcal{A}_0\widehat{\psi} &= \big(1+\tau[(1+r)\mathbf{k}^2-2\mathbf{k}^4+\mathbf{k}^6]\big)\widehat{\psi}, \label{eq:A0} \\
\mathcal{P}\widehat{\mathbf{x}} &= \begin{bmatrix}
1 & 2\mathbf{k}^2 - \mathbf{k}^4 & 0 \\
\tau \mathbf{k}^2 & 1 - \delta\mathbf{k}^2 + \delta\mathbf{k}^4 & 0 \\
0 & \mathbf{k}^2 & 1
\end{bmatrix}\begin{pmatrix}
\widehat{\mu} \\ \widehat{\psi} \\ \widehat{\omega}
\end{pmatrix},\\
\mathcal{P}_\text{CH}\widehat{\mathbf{x}} &= \begin{bmatrix}
1 & -2\sqrt{\delta}\mathbf{k}^2 - \mathbf{k}^4 & 0 \\
\tau \mathbf{k}^2 & 1 + (\delta+2\sqrt{\delta})\mathbf{k}^2 + \delta\mathbf{k}^4 & 0 \\
0 & \mathbf{k}^2 & 1
\end{bmatrix}\begin{pmatrix}
\widehat{\mu} \\ \widehat{\psi} \\ \widehat{\omega}
\end{pmatrix}.
\end{align}

In Figure \ref{fig:eigenvalues} the eigenvalue symbol curves of $\mathcal{A}$ restricted to a bounded range of frequencies, together with the distribution of eigenvalues of an assembled matrix\footnote{Since a finite element mass-matrix has eigenvalues far from the continuous eigenvalue 1, depending on the finite elements used, the grid size and the connectivity of the mesh, the overall spectrum of $\mathbf{A}$ (that contains mass-matrices on the diagonal) is shifted on the real axis. In order to compare the continuous and the discrete spectrum, we have therefore considered the diagonal preconditioned matrix $\widehat{\mathbf{A}}=\operatorname{diag}(\mathbf{A})^{-1}\mathbf{A}$ that is a blockwise diagonal scaling by the inverse of the diagonal of mass-matrices.} $\mathbf{A}$, using quadratic finite elements on a periodic tetrahedral mesh with grid size $h=\pi/4$, is shown.
The qualitative distribution of the eigenvalues is similar for symbol curves and assembled matrices, and changes as the timestep width increases. 

\begin{figure}[ht]
\begin{center}
\begin{tabular}{cc}
\multicolumn{2}{c}{\includegraphics{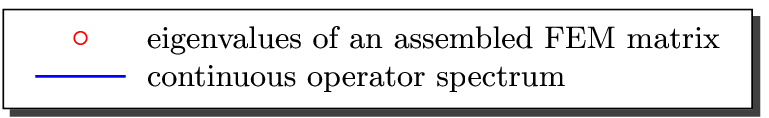}} \\
\includegraphics{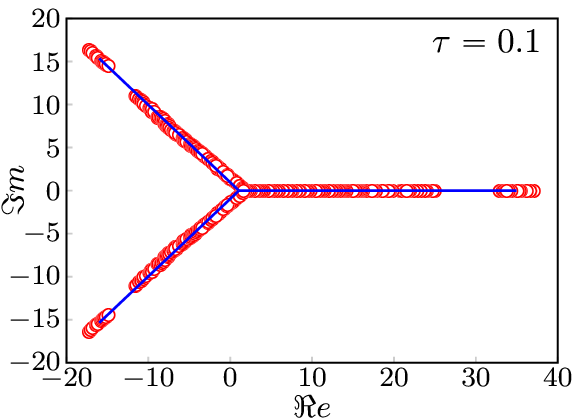} & \includegraphics{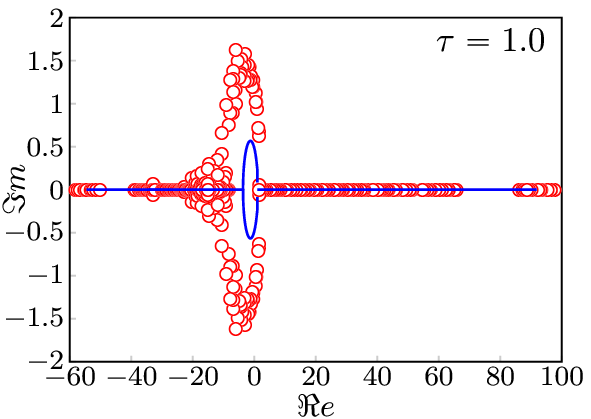} 
\end{tabular} 
\caption{Eigenvalues of the diagonally preconditioned finite element matrix $\widehat{\mathbf{A}}=\operatorname{diag}(\mathbf{A})^{-1}\mathbf{A}$, i.e. a discretization of the continuous operator $\mathbb{A}$ multiplied with the inverse of its diagonal, and the three eigenvalues of the symbol $\mathcal{A}$ visualized as restricted symbol curves. Left: spectrum for timestep width $\tau=0.1$, Right: spectrum for timestep width $\tau=1$}\label{fig:eigenvalues}
\end{center}
\end{figure}

For small $\tau$, the origin is excluded by the Y-shape profile of the spectrum. Increasing $\tau$ leads to a surrounding of the origin. This does not necessarily imply a bad convergence behavior. 

\subsection{Critical timestep width}
For larger timestep widths $\tau$ we can even get the eigenvalue zero in the continuous spectrum, i.e. the time-discretization gets unstable. In the following theorem this is analyzed in detail and a modification is proposed, that shifts this critical timestep width limit. This modification will be used in the rest of the paper

\begin{lem}\label{lem:timestep_stability}
Let $\mathcal{A}_0$ be given as in \eqref{eq:A0} and $r < 0$ then the spectrum $\sigma(\mathcal{A}_0)$ contains zero in case of the critical timestep width
\begin{equation}\label{eq:critical_timestep_width}
\tau \geq \tau^\ast := \frac{27}{2(\sqrt{\alpha} - 1)(\sqrt{\alpha} + 2)^2},
\end{equation}
with $\alpha = 1-3r$. 

Let $\bar{\psi}\in\mathbb{R}$. The spectrum of the modified operator $\hat{\mathcal{A}}_0$, given by
\begin{equation}\label{eq:A0_mod}\hat{\mathcal{A}}_0 := \big(1+\tau[(3\bar{\psi}^2+1+r)\mathbf{k}^2-2\mathbf{k}^4+\mathbf{k}^6]\big),\end{equation}
contains zero only in the case of $\bar{\psi}^2 < \frac{-r}{3}$. Then the critical timestep width is given by \eqref{eq:critical_timestep_width} with $\alpha =1-3r-9\bar{\psi}^2$.
\end{lem}
\begin{rem}
The modified operator $\hat{\mathcal{A}}_0$ can be derived by linearizing $\psi^3$ around a constant reference density $\bar{\psi}$: 
\[\psi^3\approx 3\bar{\psi}^2\psi - 2\bar{\psi}^3.\]
Adding this as an approximation of the non-linear term to the system \eqref{eq:A0} leads to the operator symbol \eqref{eq:A0_mod}.
\end{rem}

\begin{rem}
If we take $\bar{\psi}$ as the constant mean value of $\psi$ over $\Omega$, with $r$ has the physical meaning of an undercooling of the system, then the relation $|\bar{\psi}| = \sqrt{-r/3}$ is related to the solid-liquid transition in the phase-diagram of the PFC model, i.e. $|\bar{\psi}| > \sqrt{-r/3}$ leads to stable constant solutions, interpreted as liquid phase, and $|\bar{\psi}| < \sqrt{-r/3}$ leads to an instability of the constant phase, interpreted as crystalline state. An analysis of the stability condition can be found in \cite{Cheng2008, Elder2004} upon other.
\end{rem}

\begin{rem}
In \cite{Wise2009, Hu2009} an unconditionally stable discretization is provided that changes the structure of the matrix, i.e. the negative $2\mathbf{k}^4$ term is moved to the right-hand side of the equation. In order to analyze also the Rosenbrock scheme we can not take the same modification into account. The modification shown here is a bit similar to the stabilization proposed in \cite{Elsey2013}, but the authors have added a higher order term $C\mathbf{k}^4$ instead of the lower order term $C'\mathbf{k}^2$ in \eqref{eq:A0_mod}.
\end{rem}
\begin{proof} (Lemma \ref{lem:timestep_stability}). We analyze the eigenvalues of $\hat{\mathcal{A}}_0$ and get the eigenvalues of $\mathcal{A}_0$ as a special case for $\bar{\psi}=0$. The eigenvalue symbol $\hat{\mathcal{A}}_0$ gets zero whenever
\[\hat{\mathcal{A}}_0 = 0 \;\Leftrightarrow\; \tau = \frac{-1}{(3\bar{\psi}^2 + 1 + r)\mathbf{k}^2-2\mathbf{k}^4+\mathbf{k}^6}.\]
The minimal $\tau\in\mathbb{R}_{>0}$, denoted by $\tau^\ast$, that fulfils this equality is reached at
\[\mathbf{k}^2 = \frac{2}{3} + \frac{1}{3}\sqrt{1-3r-9\bar{\psi}^2} =: \frac{2}{3} + \frac{1}{3}\sqrt{\alpha}.\]
Inserting this into $\tau$ gives
\[\tau^\ast := \frac{27}{2(\sqrt{\alpha} - 1)(\sqrt{\alpha} + 2)^2}.\]
We have $\alpha \geq 0$ for $|\bar{\psi}| \geq \frac{1}{3}\sqrt{1-3r}$ and 
$\tau^\ast > 0 \;\Leftrightarrow\; \alpha > 1 \;\Leftrightarrow\; \bar{\psi}^2 > \frac{-r}{3}$
by simple algebraic calculations.
\end{proof}

On account of this zero eigenvalue, we restrict the spectral analysis to small timestep widths $\tau$. For $r=-0.35$, as in the numerical examples below, we get the timestep width bound $0<\tau<\tau^\ast\approx 2.6548$ for the operator $\mathcal{A}_0$ and with $\bar{\psi} = -0.34$ we have $0<\tau<\tau^\ast\approx 312.25$ for $\hat{\mathcal{A}}_0$, hence a much larger upper bound. In the following, we will use the modified symbols for all further calculations, i.e.
\begin{equation}\label{eq:A_mod}
\hat{\mathcal{A}} = \begin{bmatrix}
1 & -(3\bar{\psi}^2 + 1+r) +2\mathbf{k}^2 & \mathbf{k}^2 \\
\tau \mathbf{k}^2 & 1 & 0 \\
0 & \mathbf{k}^2 & 1
\end{bmatrix},
\end{equation}
and remove the hat symbol for simplicity, i.e. $\hat{\mathcal{A}}\rightarrow\mathcal{A},\;\hat{\mathcal{A}}_0\rightarrow\mathcal{A}_0$.

Calculating the eigenvalues of the preconditioner symbol $\mathcal{Q}:=\mathcal{A}\mathcal{P}^{-1}$, respective $\mathcal{Q}_\text{CH}:=\mathcal{A}\mathcal{P}_\text{CH}^{-1}$, directly gives the sets
\begin{align}
\sigma(\mathcal{Q}) &=\left\{1,1, \frac{\tau(\mathbf{k}^6 - 2\mathbf{k}^4 + (3\bar{\psi}^2 + 1+r)\mathbf{k}^2) + 1}{\tau\mathbf{k}^6 + (\sqrt{\tau} - 2\tau)\mathbf{k}^4 - \sqrt{\tau}\mathbf{k}^2 + 1}\; \Big| \;\mathbf{k}\in\mathbb{R}^m\right\} \label{eq:eigenvalues_B}\\
\sigma(\mathcal{Q}_\text{CH})&=\left\{1,1, \frac{\tau(\mathbf{k}^6 - 2\mathbf{k}^4 + (3\bar{\psi}^2 + 1+r)\mathbf{k}^2) + 1}{\tau\mathbf{k}^6 + (\sqrt{\tau} + 2\tau^{3/4})\mathbf{k}^4 + (\sqrt{\tau} + 2\tau^{1/4})\mathbf{k}^2 + 1}\; \Big| \;\mathbf{k}\in\mathbb{R}^m\right\}\label{eq:eigenvalues_B_ch}
\end{align}
with values all on the real axis (for $\tau > 0$). Similar to the analysis of $\mathcal{A}_0$ we get a critical timestep width, i.e. eigenvalues zero, for $\tau\geq\tau^\ast$. The denominator of the third eigenvalue of $\sigma(\mathcal{Q}_\text{CH})$ is strictly positive, but the denominator of $\sigma(\mathcal{Q})$ can reach zero. This would lead to bad convergence behavior, since divergence of this eigenvalue would lead to divergence of the asymptotic convergence factor in \eqref{eq:estimate_interval}. 

The critical timestep width, denoted by $\tau^\natural$, that allows a denominator with value zero is given by
$\tau = (-\mathbf{k}^4 + 2\mathbf{k}^2)^{-2}$
that is minimal positive for $|\mathbf{k}|=1$ and gives $\tau^\natural = 1$. Thus, for the preconditioner $\mathcal{P}$ we have to restrict the timestep width to $\tau\in(0,\tau^\natural)\subset(0,\tau^\ast)$. This restriction is not necessary for the preconditioner $\mathcal{P}_\text{CH}$.

\subsection{The asymptotic convergence factor}
Since the third eigenvalue of $\mathcal{Q}$ in \eqref{eq:eigenvalues_B}, respective $\mathcal{Q}_\text{CH}$ in \eqref{eq:eigenvalues_B_ch}, is a real interval $\subset\mathbb{R}_+$, for $\tau$ in the feasible range $(0,\min(\tau^\natural,\tau^\ast))$, we can use formula \eqref{eq:estimate_interval} to estimate an asymptotic convergence factor for the lower bound on $\min_{p\in\Pi_k}\|p(\mathcal{Q}_\ast)\|$. For fixed $r=-0.35$ and various values of $\bar{\psi}$ minimum and maximum of \eqref{eq:eigenvalues_B} and  \eqref{eq:eigenvalues_B_ch} are calculated numerically. Formula \eqref{eq:estimate_interval} thus gives the corresponding estimated asymptotic convergence factor (see left plot of Figure \ref{fig:asymptotic_convergence_factor}). For step widths $\tau$ less than 1 we have the lowest convergence factor for the operator $\mathcal{Q}$ and the largest for the original operator $\mathcal{A}_0$. The operator $\mathcal{Q}_\text{CH}$ has slightly greater convergence factor than $\mathcal{Q}$ but is more stable with respect to an increase in timestep width.

The stabilization term $3\bar{\psi}^2$ added to $\mathcal{A}_0$ and $\mathcal{A}$ in \eqref{eq:A0_mod}, respective \eqref{eq:A_mod}, influences the convergence factor of $\mathcal{A}_0$ and $\mathcal{Q}$ only slightly, but the convergence factor of $\mathcal{Q}_\text{CH}$ is improved a lot, i.e. the critical timestep width is shifted towards positive infinity.

The upper bounds on the convergence factor are analyzed in the next Section.

\begin{figure}[ht]
\begin{center}
\begin{tabularx}{\textwidth}{Xc}
\parbox[b]{\hsize}{\includegraphics{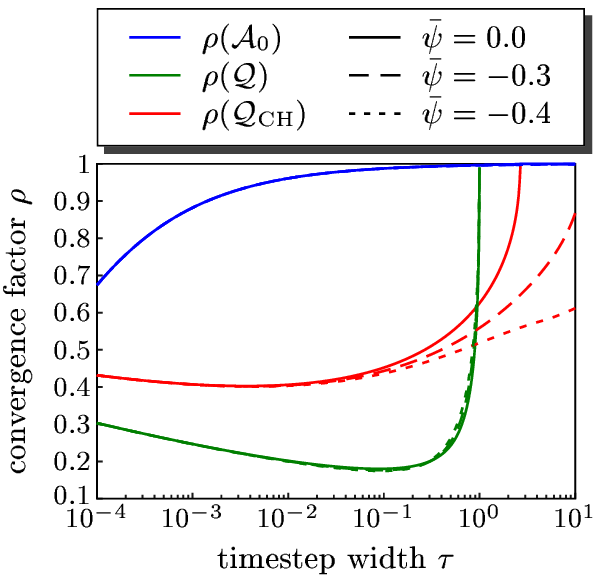}} & \includegraphics{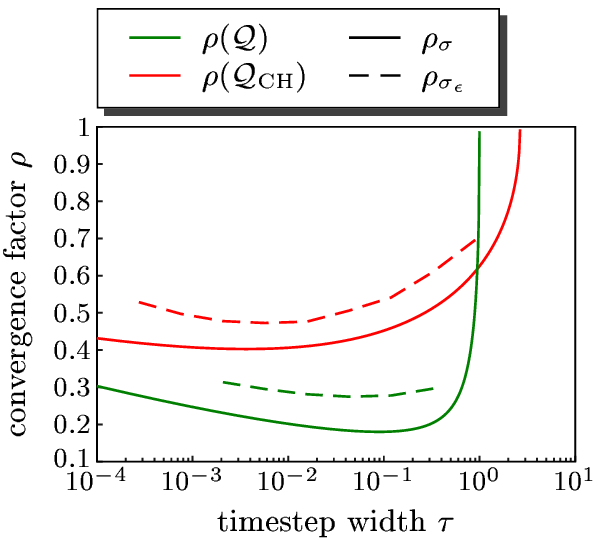} \\ 
\end{tabularx} 
\caption{Left: Asymptotic convergence factor for operators $\mathcal{A}_0,\mathcal{Q}$ and $\mathcal{Q}_\text{CH}$. In dashed lines, the dependence on the mean density $\bar{\psi}$-modification \eqref{eq:A0_mod} is shown. Right: Comparison of the convergence factor related to the spectrum and $\epsilon$-pseudospectrum is shown. This corresponds to lower and upper bounds of the actual asymptotic convergence factor.}\label{fig:asymptotic_convergence_factor}
\end{center}
\end{figure}

\subsection{Analysis of the pseudospectrum}
As it can be seen by simple calculations, the symbol $\mathcal{A}$ is non-normal:
\[\big(\mathcal{A}^\top\mathcal{A} - \mathcal{A}\mathcal{A}^\top\big)_{2,0} = \mathbf{k}^2\big((3\bar{\psi}^2+1+r) - 2\mathbf{k}^2\big)\neq 0\]
for a matrix entry at row 2 and column 0. For slightly more complex calculations it can be shown, that also $\mathcal{Q}:=\mathcal{A}\mathcal{P}^{-1}$ and $\mathcal{Q}_\text{CH}:=\mathcal{A}\mathcal{P}_\text{CH}^{-1}$ are non-normal:
\[\big(\mathcal{Q}^\top\mathcal{Q} - \mathcal{Q}\mathcal{Q}^\top\big)_{2,2} = \big(\mathcal{Q}_\text{CH}^\top\mathcal{Q}_\text{CH} - \mathcal{Q}_\text{CH}\mathcal{Q}_\text{CH}^\top\big)_{2,2} = \mathbf{k}^4\neq 0.\]

For non-normal matrices we have to analyze the $\epsilon$-pseudospectrum in order to get an estimate of convergence bounds for the GMRES method, as pointed out in Section \ref{sec:convergence_prefactor}. 


Using the Matlab Toolbox \texttt{Eigtool} provided by \cite{Eigtool} we can calculate the pseudospectra $\sigma_\epsilon$ and approximations $\Gamma_\epsilon$ of its boundaries with single closed Jordan curves for all wave-numbers $k_i\in[0,k_{\max}]$. The maximal frequency used in the calculations is related to the grid size $h$ of the corresponding triangulation, $k_{\max}=\frac{\pi}{h}$. In all the numerical examples below, we have used a grid size $h=\frac{\pi}{4}$ that can resolve all the structures sufficiently well, thus we get $k_{\max}=4$ that leads to $|\mathbf{k}|_{\max}=\sqrt{m}k_{\max}$.

The $\epsilon$-pseudospectrum of $\mathcal{Q}$ and $\mathcal{Q}_\text{CH}$ can be seen in Figure \ref{fig:pseudospectrum} for various values of $\epsilon$. The pseudospectrum of $\mathcal{Q}_\text{CH}$ gets closer to the origin than that of $\mathcal{Q}$, since the eigenvalues get closer to the origin as well. The overall structure of the pseudospectra is very similar.

\begin{figure}[ht]
\begin{center}
\includegraphics{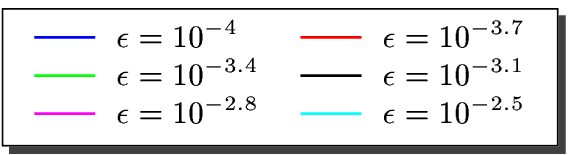}
\begin{tabular}{cc}
\includegraphics{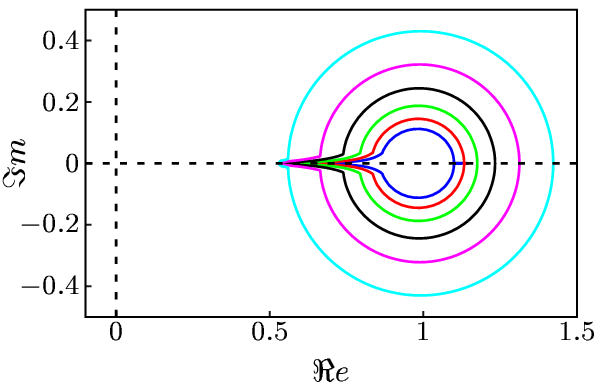} & \includegraphics{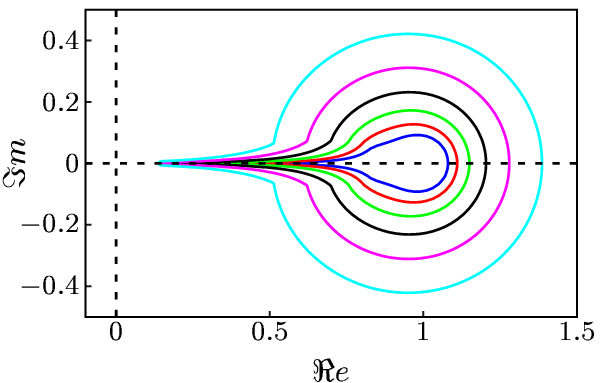} \\ 
\end{tabular} 
\caption{$\epsilon$-Pseudospectra of the preconditioned matrices for various values of $\epsilon$. Left: $\sigma_{\epsilon}(\mathcal{Q})$, Right: $\sigma_{\epsilon}(\mathcal{Q_\text{CH}})$, for $\tau=0.1$ and $\mathbf{k}$ in the restricted range $[0,|\mathbf{k}|_{\max}]$. The dashed lines correspond to the zero-axis and indicate the origin. }\label{fig:pseudospectrum}
\end{center}
\end{figure}

For the convergence factor corresponding to the pseudospectra, we compute the inverse conformal map $\Psi=\Phi^{-1}$ of the exterior of the unit disk to the exterior of a polygon $\mathcal{S}_0$ approximating the set $\mathcal{S}=\sigma_\epsilon$, by the Schwarz-Christoffel formula, using the SC \textsc{Matlab} Toolbox \cite{Driscoll1996,Driscoll1998b}. A visualization of the inverse map $\Psi$ for the pseudospectrum of $\mathcal{A}$ and $\mathcal{Q}$ can be found in Figure \ref{fig:conformal_map}.

\begin{figure}[ht]
\begin{center}
\begin{tabular}{cc}
\includegraphics{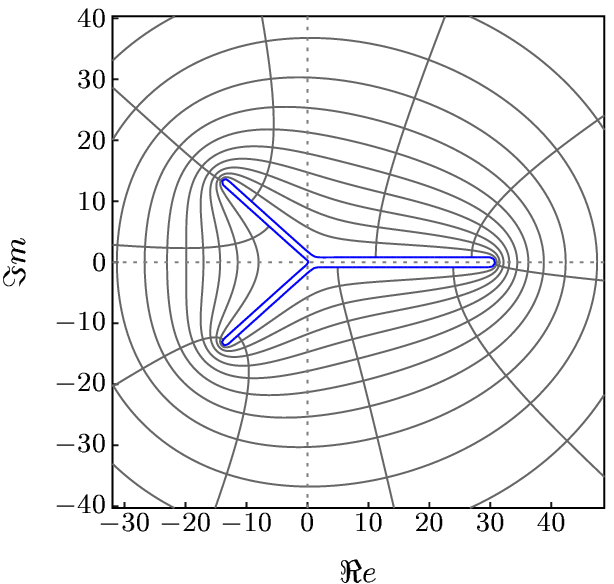} & \includegraphics{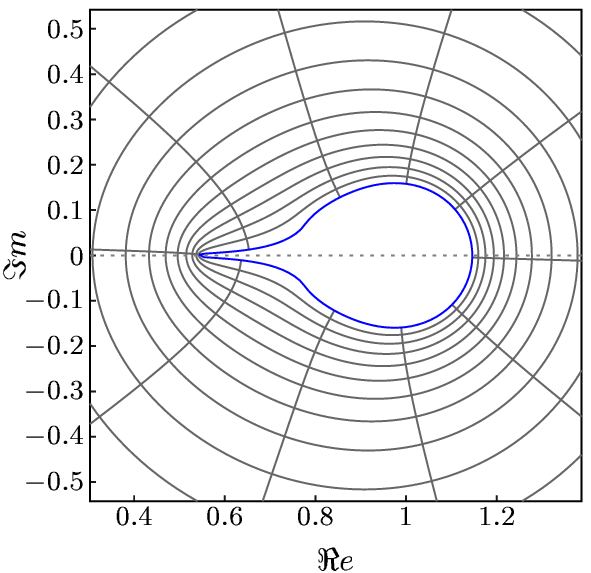} \\ 
\end{tabular} 
\caption{Inverse conformal map $\Psi$ of the unit disk to the exterior of a polygon enclosing the $\epsilon$-pseudospectrum of $\mathcal{A}$, respective $\mathcal{Q}$ for one $\epsilon$ and a restricted range of frequencies. For both plots the $\epsilon$ value is chosen small enough to have 0 in the exterior of the peudospectrum. }\label{fig:conformal_map}
\end{center}
\end{figure}

Evaluating the asymptotic convergence factor depending on the $\epsilon$-pseudospectrum of the matrices is visualized in Figure \ref{fig:prefactor_asymptotic_convergence_factor}. The calculation is performed for fixed timestep width $\tau=0.1$ and parameters $r=-0.35$ and $\bar{\psi}=0$.

Increasing $\epsilon$ increases the radius of the sphere like shape around the point $1$. For a simple disc the convergence factor is proportional to the radius (see \cite{Driscoll1998}), thus we find increasing convergence factors also for our disc with the tooth. When $\epsilon$ gets too large the pseudospectrum may contain the origin that would lead to useless convergence bounds, since then $\rho_{\sigma_\epsilon}>1$ in \eqref{eq:convergence_factor_S}. If $\epsilon$ gets too small, the convergence constant $C$ in \eqref{eq:upper_bound_min_max} is growing rapidly, since $|\Gamma_\epsilon|$ is bounded from below by the eigenvalue interval length, i.e. $|\Gamma_\epsilon|\leq 2(\beta - \alpha) = 2\big((\max(\sigma(\mathcal{Q})), \min(\sigma(\mathcal{Q}))\big)$, and we divide by $\epsilon$. Thus, the estimates also are not meaningful in the limit $\epsilon\rightarrow 0$.

An evaluation of the constant $C$ for various values $\epsilon$ can be found in the left plot of Figure \ref{fig:prefactor_asymptotic_convergence_factor}. It is a log-log plot with constants in the range $[10^2, 10^4]$.

For all $\epsilon>0$ the upper bound \eqref{eq:upper_bound_min_max} is valid, so we have chosen $\epsilon=10^{-3}$ and plotted the resulting estimated asymptotic convergence factor in relation to the lower bound, i.e. the convergence factor corresponding to the spectrum of the matrices, in the right plot of Figure \ref{fig:asymptotic_convergence_factor}. The upper bound is just slightly above the lower bound. Thus, we have a convergence factor for $\mathcal{Q}$ that is in the range $0.2-0.3$ (for $\tau=0.1$ and $r=-0.35$) and for the matrix $\mathcal{Q}_\text{CH}$ in the range $0.45-0.55$, that is much lower than the lower bound of the convergence factor of $\mathcal{A}_0$ (approximately $0.99$). 

So both, the preconditioner $\mathcal{P}$ and $\mathcal{P}_\text{CH}$ improves the asymptotic convergence factor a lot and we expect fast convergence also in the case of discretized matrices.
\begin{figure}[ht]
\begin{center}
\begin{tabular}{cc}
\includegraphics{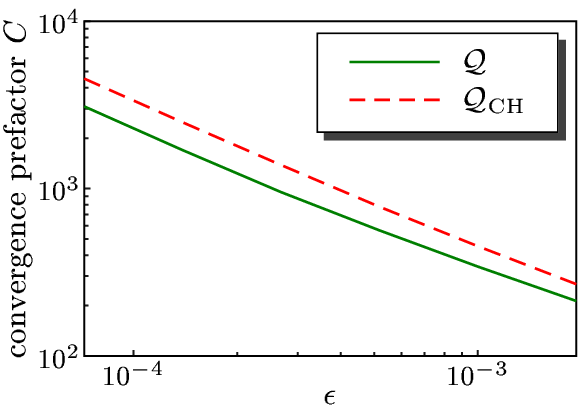} & \includegraphics{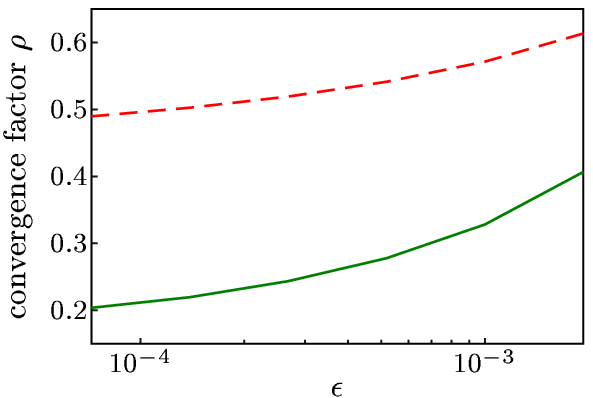} 
\end{tabular} 
\caption{Left: Estimated convergence prefactor $C:=\frac{|\Gamma_\epsilon|}{2\pi\epsilon}$ for the matrix $\mathcal{Q}=\mathcal{A}\mathcal{P}^{-1}$ and  $\mathcal{Q}_\text{CH}=\mathcal{A}\mathcal{P}_\text{CH}^{-1}$ with $r=-0.35$ and $\tau=0.12$, plotted in logarithmic scale for $\epsilon$ and $C$. Right: Estimated asymptotic convergence factor $\rho_{\sigma_\epsilon(\mathcal{Q})}=|\Phi_{\sigma_\epsilon(\mathcal{Q})}(0)|^{-1}$ and $\rho_{\sigma_\epsilon(\mathcal{Q}_\text{CH})}$ analogously. The legend in the left plot is valid also for the right plot.}\label{fig:prefactor_asymptotic_convergence_factor}
\end{center}
\end{figure}

\section{Numerical studies}\label{seq:numerical_examples}
We now demonstrate the properties of the preconditioner numerically. We consider a simple crystallization problem in 2D and 3D, starting with an initial grain in a corner of a rectangular domain. 
The solution of the PFC equation in the crystalline phase is a periodic wave-like field with specific wave length and amplitude. In \cite{Elder2004,Jaatinen2010} a single mode approximation for the PFC equation in 2D and 3D is provided. These approximations show a wave length of $\mathit{d}:=4\pi/\sqrt{3}$, corresponding to a lattice spacing in a hexagonal crystal in 2D and a body-centered cubic (BCC) crystal in 3D. We define the domain $\Omega$ as a rectangle/cube with edge length, a multiple of the lattice spacing: $\Omega=[N\cdot \mathit{d}]^{2,3}$, with $N\in\mathbb{N}_{>0}$. Discretizing one wave with 10 gridpoints leads to a sufficient resolution. Our grid size therefore is $h=\frac{d}{10}\approx\frac{\pi}{4}$ throughout the numerical calculations. We use regular simplicial meshes, with $h$ corresponding to the length of an edge of a simplex for linear elements and twice its length for quadratic elements to guarantee the same number of degrees of freedom (DOFs) within one wave.

\subsection{General problem setting and results}
As system parameters we have chosen values corresponding to a coexistence of liquid and crystalline phases: 2D ($\bar\psi = -0.35$, $r=-0.35$), 3D ($\bar\psi = -0.34$, $r=-0.3$). Both parameter sets are stable for large timestep widths, with respect to Lemma \ref{lem:timestep_stability}. In Figure \ref{fig:coexistence} snapshots of the coexistence regime of liquid and crystal is shown for two and three dimensional calculations.

\begin{figure}[ht]
\begin{center}
\includegraphics[width=.28\linewidth]{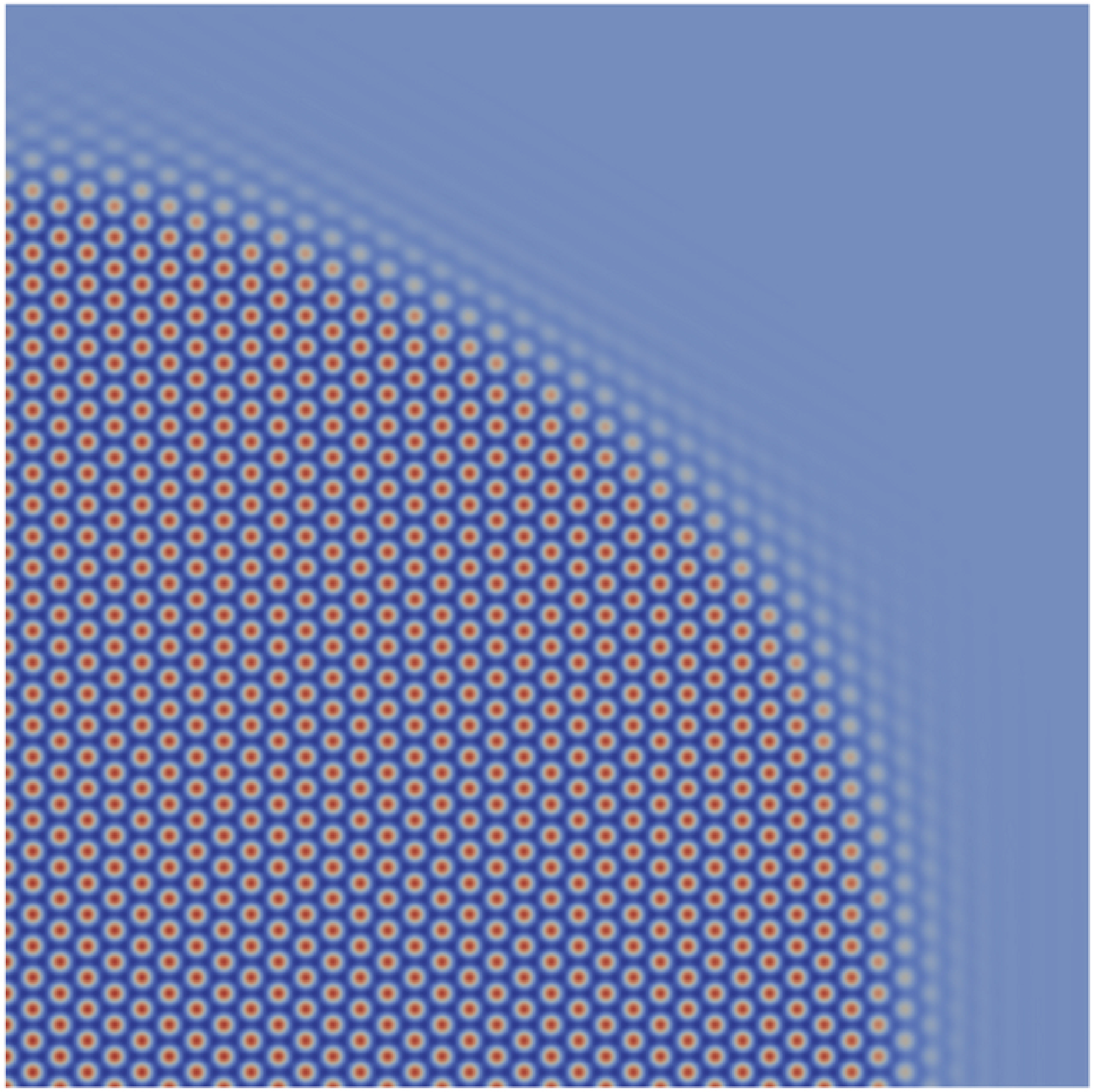}  
\includegraphics[width=.31\linewidth]{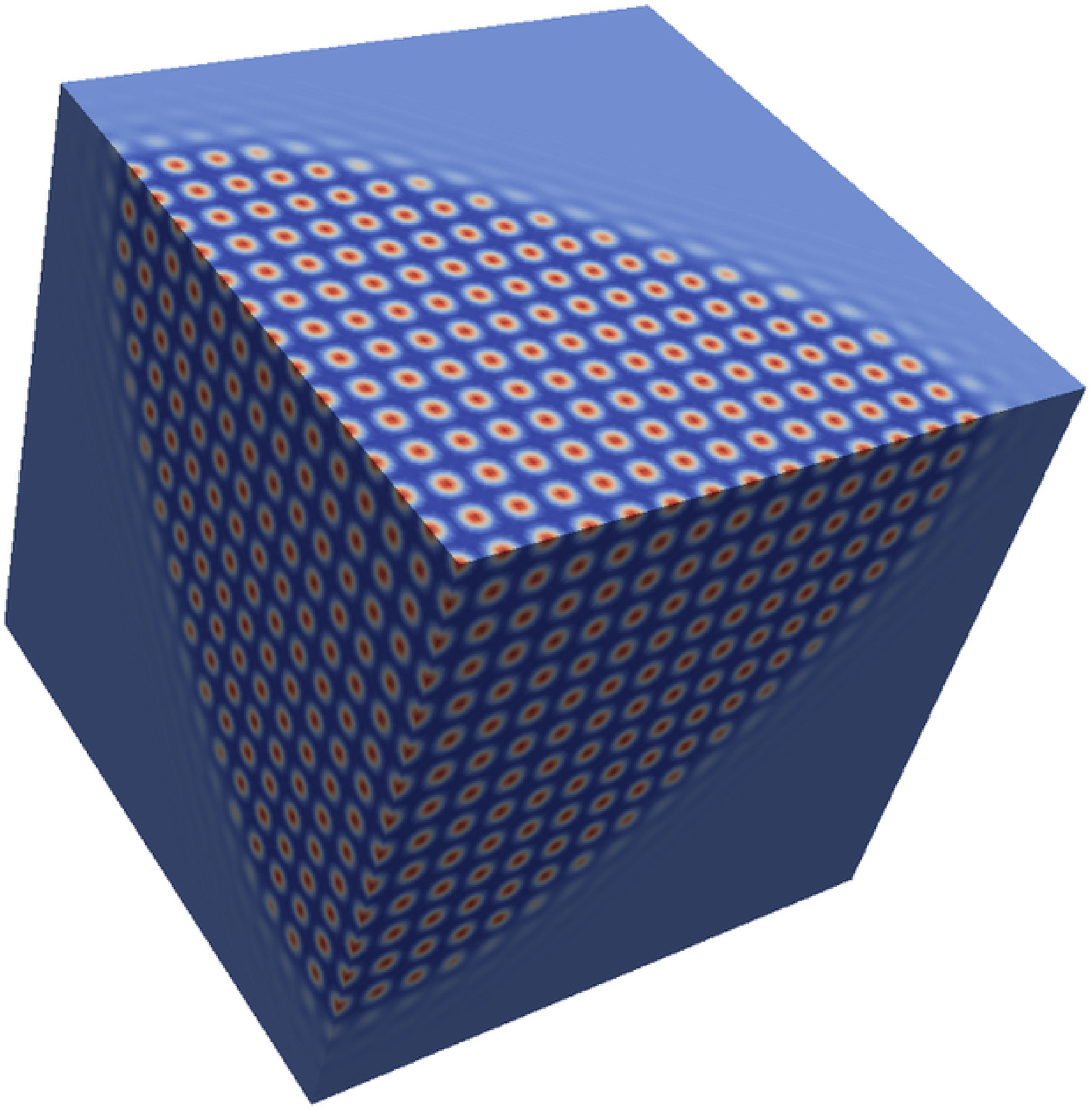} 
\caption{Intermediate state of growing crystal, starting from one corner of the domain. Shown is the order parameter field $\psi$. Left: $\Omega=[20d]^2$, number of DOFs: 263,169, calculated on 1 processor; Right: $\Omega=[12d]^3$, number of DOFs: 101,255,427, calculated on 3.456 processors. }\label{fig:coexistence} 
\end{center}
\end{figure}

In the steps (1), (2), (3) and (5) of the preconditioner solution procedure, linear systems have to be solved. For this task we have chosen iterative solvers with standard preconditioners and parameters as listed in Table \ref{tbl:inner_solvers}. The PFC equation is implemented in the finite element framework AMDiS \cite{AMDiS, AMDiS_parallel} using the linear algebra backend MTL4 \cite{MTL4, MTL4-cuda, AMDiS-MTL4} in sequential calculations and PETSc \cite{PETSc} for parallel calculation for the block-preconditioner \eqref{eq:pfc-preconditioner} and the inner iterative solvers. As outer solver, a FGMRES method is used with restart parameter 30 and modified Gram-Schmidt orthogonalization procedure. The spatial discretization is done using Lagrange elements of polynomial degree $p =1,2$ and as time discretization the implicit Euler or the described Rosenbrock scheme is used.  

\begin{table}[ht]
\begin{center}
\begin{tabular}{l|l|l|l|l}
precon. steps & matrix & solver & precond. & rel. tolerance \\ 
\hline 
(1),(5) & $\mathbf{M}$ & PCG & diag & $10^{-3}$\\ 
(2) & $\mathbf{M} + \delta\mathbf{K}$ & PCG & diag & $10^{-3}$ \\ 
(3.1), (3.2) & $\mathbf{M} + \sqrt{\delta}\mathbf{K}$ & PCG & diag & $10^{-3}$\\ 
(3) & $\mathbf{M} - 2\delta\mathbf{K}+\delta\mathbf{K}\mathbf{M}_D^{-1}\mathbf{K}$ & PCG & diag & 20 (iter.)\\ 
\end{tabular} 
\end{center}
\caption{Parameters for the inner solvers of the preconditioner with Cahn-Hilliard approximation $\mathbf{S}_\text{CH}$ and in the last line for the diagonal approximation $\mathbf{S}_D$ of the matrix $\mathbf{S}$. `PCG' is the shortcut for \textit{preconditioned conjugate gradient} method. The preconditioner named `diag' indicates a Jacobi preconditioner. We have solved each inner system up to a relative solver tolerance given in the last column of the table. Only in the case of the matrix $\mathbf{S}_D$ it is more efficient to use a fixed number of iteration.}\label{tbl:inner_solvers}
\end{table}

The first numerical test compares a PFC system solved without a preconditioner to a system solved with the developed preconditioner. In Figure \ref{fig:precon_no-precon} the relative residual in the first timestep of a small 2D system is visualized. For increasing timestep widths the FGMRES solver without preconditioner (dashed lines) shows a dramatic increase of the number of iterations up to a nearly stagnating curve for timestep widths greater than 0.5. On the other hand, we see in solid lines the preconditioned solution procedure that is much less influenced by the timestep widths and reaches the final residual within 20-30 iterations. A detailed study of the influence of the timestep width can be found below. For larger systems, respective systems in 3D, we get nearly no convergence for the non-preconditioned iterations. 

\begin{figure}[ht]
\begin{center}
\includegraphics{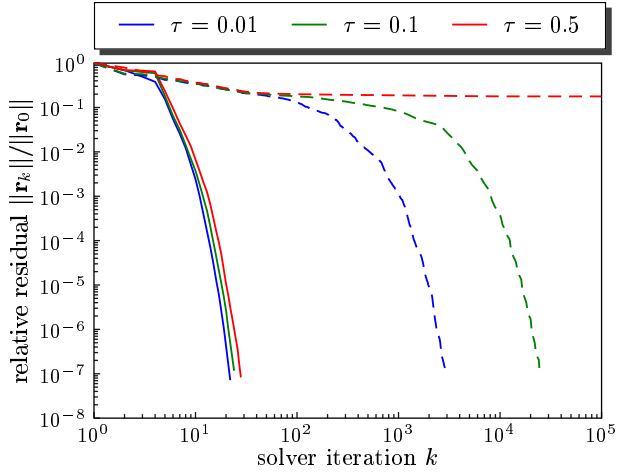}
\caption{Relative residual of the solver iterations. Solid lines show preconditioned solver iterations and dashed lines iterations without a preconditioner. The systems is $\Omega=[d]^2,\;h=\frac{\pi}{4}$.}\label{fig:precon_no-precon}
\end{center}
\end{figure}

Next, we consider the solution procedure of the sub-problems in detail. Table \ref{tbl:direct_iterative_solvers} shows a comparison between an iterative \textit{preconditioned conjugate gradient} method (PCG) and a direct solver, where the factorization is calculated once for each sub-problem matrix per timestep. The number of outer iterations increases, when we use iterative inner solvers, but the overall solution time decreases since a few PCG steps are faster than the application of a LU-factorization to the Krylov vectors. This holds true in 2D and 3D for polynomial degree 1 and 2 of the Lagrange basis functions. 

\begin{table}[ht]
\begin{center}
\begin{tabular}{c||c|c|c||c|c}
& & \multicolumn{2}{c||}{direct} & \multicolumn{2}{c}{iterative}  \\ 
dim. $m$ & poly. degree $p$ & time [sec] & \#iterations & time [sec] & \#iterations \\ 
\hline 
\multirow{2}{*}{2D} & 1 & 6.11 & 14 & 1.65 & 14\\ 
& 2 & 5.89 & 14 & 3.04 & 15\\ 
\hline
\multirow{2}{*}{3D} & 1 & 41.72 & 17 & 3.36 & 18\\ 
& 2 & 35.24 & 17 & 9.62 & 19\\ 
\end{tabular} 
\end{center}
\caption{Comparison of the number of iterations and time to solve the linear system averaged over 20 timesteps for the preconditioner matrix $\mathbf{S}_\text{CH}$ with timestep width $\tau=0.1$. Sub-problems of the preconditioner are solved with iterative solvers as in Table \ref{tbl:inner_solvers} or with the direct solver UMFPACK. The benchmark configuration is a problem with approximately 66,000 DOFs and grid size $h=\pi/4$.}\label{tbl:direct_iterative_solvers}
\end{table}

We now compare the two proposed preconditioners regarding the same problem. Table \ref{tbl:inner_P_Pch_solvers} shows a comparison of the approximation of the sub-problem (3) by either the diagonal mass-matrix approximation $\mathbf{S}_D$ or the Cahn-Hilliard preconditioner approximation $\mathbf{S}_\text{CH}$. In all cases, the number of outer solver iterations needed to reach the relative tolerance and also the time for one outer iteration is lower for the $\mathbf{S}_\text{CH}$ approximation than for the $\mathbf{S}_D$ approximation. 

\begin{table}[ht]
\begin{center}
\begin{tabular}{c||c|c|c||c|c}
& & \multicolumn{2}{c||}{$\mathbf{S}_\text{CH}$} & \multicolumn{2}{c}{$\mathbf{S}_D$}  \\ 
dim. $m$ & poly. degree $p$ & time [sec] & \#iterations & time [sec] & \#iterations \\ 
\hline 
\multirow{2}{*}{2D} &  1 & 1.65 & 14 & 2.72 & 16\\ 
& 2 & 3.04 & 15 & 7.03 & 20\\ 
\hline 
\multirow{2}{*}{3D} & 1 & 3.36 & 18 & 8.14 & 21\\ 
& 2 & 9.62 & 19 & 81.49 & 55\\ 
\end{tabular} 
\end{center}
\caption{Comparison of the number of iterations and time to solve the linear system averaged over 20 timesteps for the preconditioner with diagonal approximation $\mathbf{S}_D$ of $\mathbf{S}$, respective Cahn-Hilliard approximation $\mathbf{S}_\text{CH}$. Sub-problems of the preconditioner are solved with iterative solvers as in Table \ref{tbl:inner_solvers}. 
The benchmark configuration is a problem with approximately 66,000 DOFs, timestep width $\tau=0.1$ and grid size of $h=\pi/4$.}\label{tbl:inner_P_Pch_solvers}
\end{table}

In the following, we thus use the preconditioned solution method with $\mathbf{S}_\text{CH}$ and PCG for the sub-problems. We analyse the dependence on the timestep width in detail and compare it with the theoretical predictions and show parallel scaling properties. 

\subsection{Influence of timestep width}
In Table \ref{tbl:timestep_change} the time to solve the linear system averaged over 20 timesteps is listed for various timestep widths $\tau$. All simulations are started from the same initial condition that is far from the stationary solution. It can be found that the solution time increases and also the number of outer solver iterations increases. In Figure \ref{fig:benchmark_tau}, this increase in solution time is visualized for various parameter sets for polynomial degree and space dimension. The behaviour corresponds to the increase in the asymptotic convergence factor (see Figure \ref{fig:asymptotic_convergence_factor}) for increasing timestep widths. 

\begin{table}[ht]
\begin{center}
\begin{tabular}{c|c|c||c|c}
 & \multicolumn{2}{c||}{2D} & \multicolumn{2}{c}{3D}  \\ 
timestep width $\tau$ & time [sec] & \#iterations & time [sec] & \#iterations \\ 
\hline 
0.01 &  2.50 & 13 & 8.01 & 17 \\ 
0.1 & 3.05 & 15 & 9.62 & 19\\ 
1.0 & 4.53 & 19 & 14.29 & 24 \\ 
10.0 & 10.81 & 47 & 34.94 & 58\\ 
\end{tabular} 
\end{center}
\caption{Comparison of time to solve the linear system averaged over 20 timesteps for various timestep widths $\tau$ for a 2D and a 3D system. The benchmark configuration is a problem with polynomial degree $p=2$ with approximately 66,000 DOFs.}\label{tbl:timestep_change}
\end{table}

We have analyzed whether a critical timestep width occurs in the two approximations of $\mathbf{S}$ (see Figure \ref{fig:tau_iterations}). The diagonal approximation $\mathbf{S}_\text{D}$ is spectrally similar to the original preconditioner $\mathbf{S}$ that has shown the critical timestep width $\tau^\natural = 1$. In the numerical calculations, $\mathbf{S}_\text{D}$ shows a critical value around the analytical value, but it varies depending on the finite element approximation of the operators. For linear Lagrange elements we see $\tau^\natural\approx 2$ and for quadratic element $\tau^\natural\approx 0.6$. The Cahn-Hilliard approximation $\mathbf{S}_\text{CH}$ does not show a timestep width, where the number of outer iterations explodes, at least in the analyzed interval $\tau\in[10^{-3}, 10^1]$. The difference in the finite element approximations is also not so pronounced as in the case of $\mathbf{S}_\text{D}$.

\begin{figure}[ht]
\begin{center}
\begin{tabular}{l}
\includegraphics{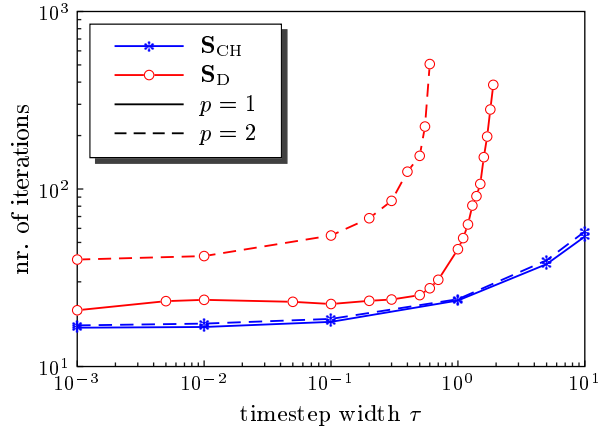}
\end{tabular} 
\caption{Increase in number of outer iterations when the timestep width increases. The red curves (curcular dots) correspond to the diagonal approximation $\mathbf{S}_\text{D}$ of $\mathbf{S}$ and the blue curves (asterisk dots) to the Cahn-Hilliard approximation $\mathbf{S}_\text{CH}$. All simulations are performed in 3D in a domain with grid size $h=\pi/4$. The solid lines correspond to simulations with polynomial degree $p=1$ and the dashed lines with polynomial degree $p=2$.}\label{fig:tau_iterations}
\end{center}
\end{figure}

\begin{figure}[ht]
\begin{center}
\includegraphics{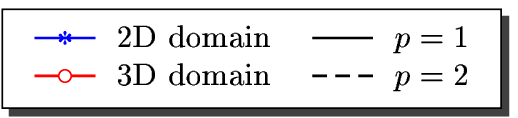}
\begin{tabular}{ll}
\includegraphics{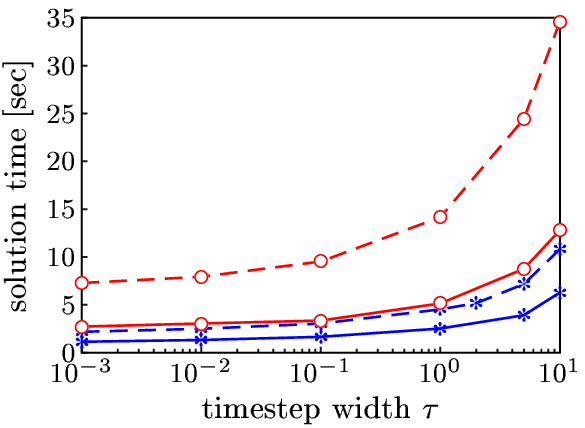} & \includegraphics{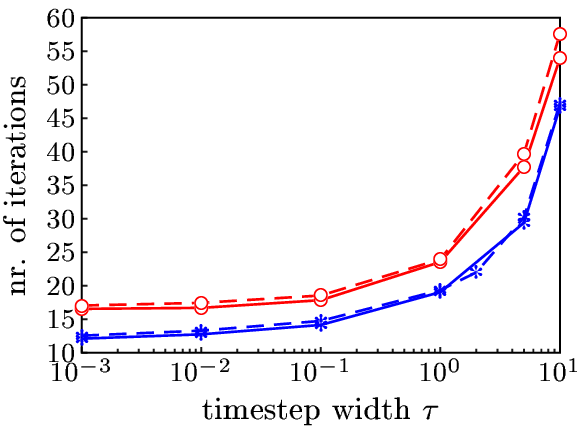} \\ 
\end{tabular} 
\caption{Left: Solution time per timestep iteration for various timestep widths $\tau$ averaged over 20 timesteps. Right: Number of outer iterations per timestep for various timestep widths $\tau$ averaged over 20 timesteps. The four curves show the dependence on dimension (2D or 3D) and on polynomial degree $p$ of the Lagrange basis functions.}\label{fig:benchmark_tau}
\end{center}
\end{figure}

While in all previous simulations an implicit Euler discretization was used, we now will demonstrate the benefit of the described Rosenbrock scheme, for which the same preconditioner is used. Adaptive time stepping becomes of relevance, especially close to the stationary solution, where the timestep width needs to be increased rapidly to reduce the energy $F(\psi)$ further. In order to allow for large timestep widths the iterative solver, respective preconditioner, must be stable with respect to an increase in this parameter. 

In Figure \ref{fig:orientation} the system setup and evolution for the Rosenbrock benchmark is shown. We use 10 initial grains randomly distributed and oriented in the domain and let the grains grow until a stable configuration emerges. When growing grains touch each other, they build grain boundaries with crystalline defects. The orientation of the final grain configuration is shown in the right plot of Figure \ref{fig:orientation} with a color coding with respect to an angle of the crystal cells relative to a reference orientation.

The time evolution of the timestep width obtained by an adaptive step size control and the evolution of the corresponding solver iterations is shown in the left plot of Figure \ref{fig:timestep_solution_time}. Small grains grow until the whole domain is covered by particles. This happens in the time interval $[0,200]$, where small timestep widths are required. From this time, the timestep width is increased a lot by the step size control since the solution is in a nearly stable state. The number of outer solver iterations increases with increasing timestep width, as expected. Timestep widths up to 18 in the time evolution are selected by the step size control and work fine with the proposed preconditioner.

In the right plot of Figure \ref{fig:timestep_solution_time}, the relation of the obtained timestep widths to the solution time is given. Increasing the timestep widths increases also the solution time, but the increasing factor is much lower than that of the increase in timestep width, i.e. the slope of the curve is much lower than 1. Thus, it is advantageous to increase the timestep widths as much as possible to obtain an overall fast solution time.

\begin{figure}[ht]
\begin{center}
\begin{tabular}{ccc}
\includegraphics[width=.28\textwidth]{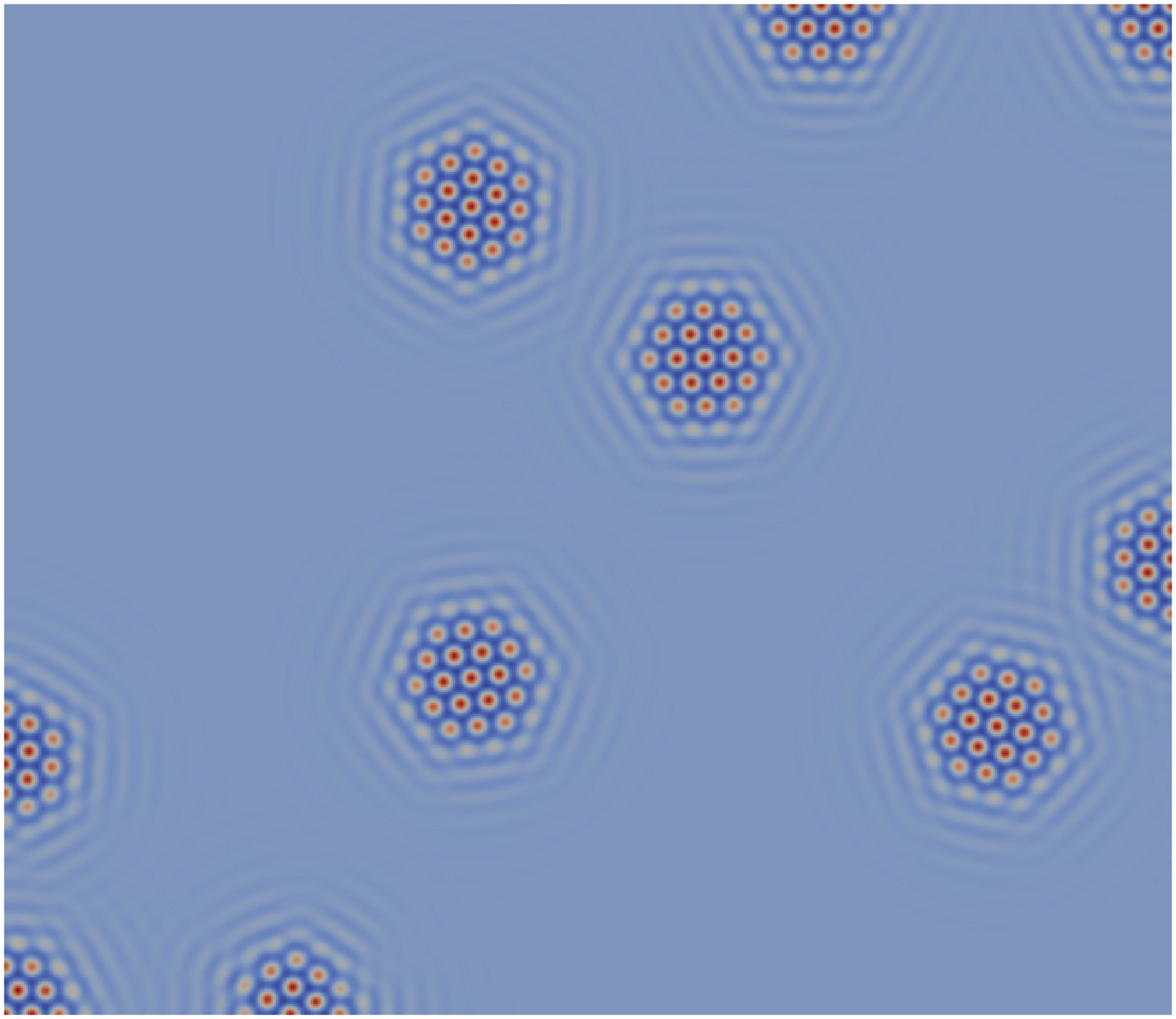} & \includegraphics[width=.28\textwidth]{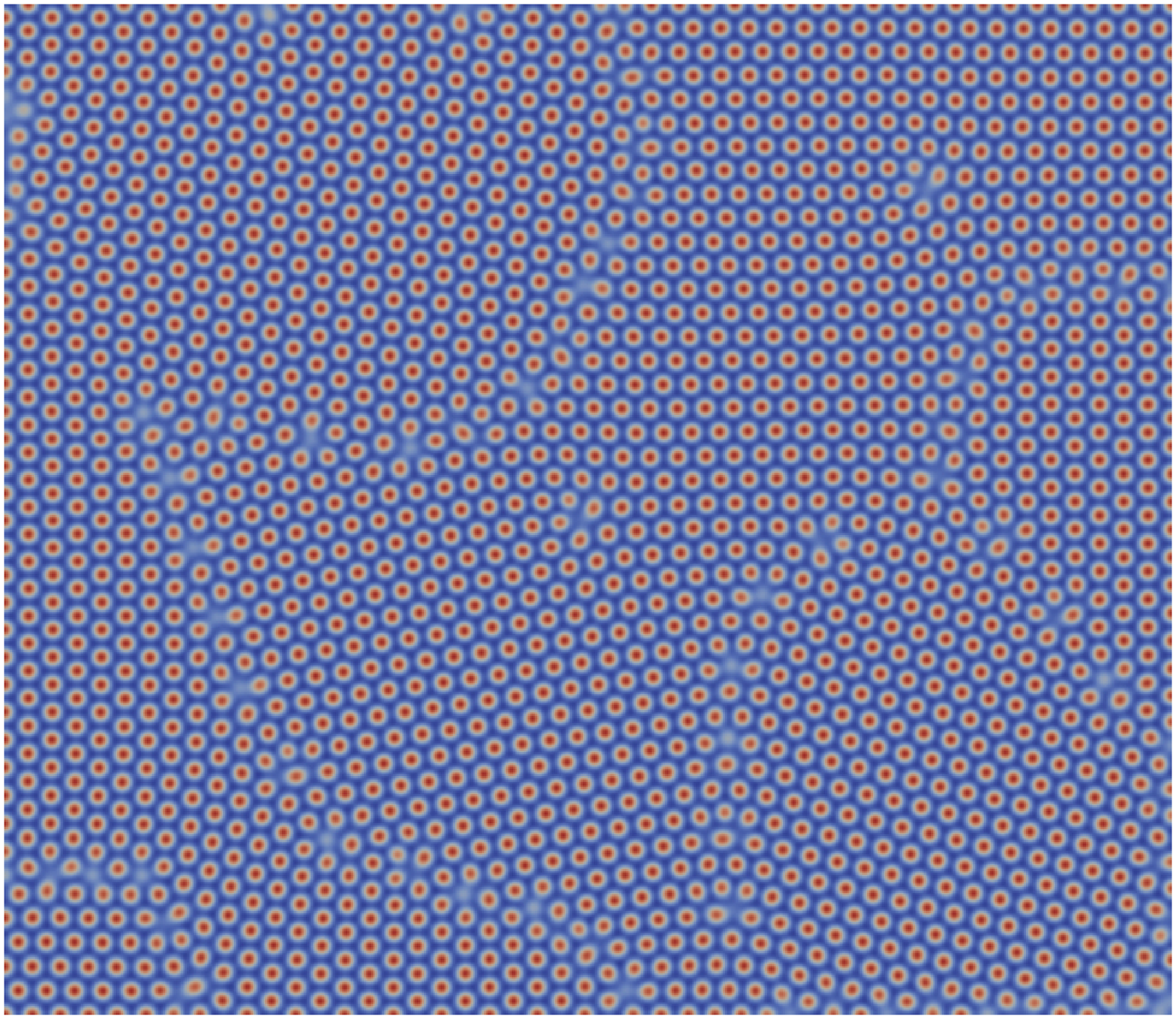} & \includegraphics[width=.28\textwidth]{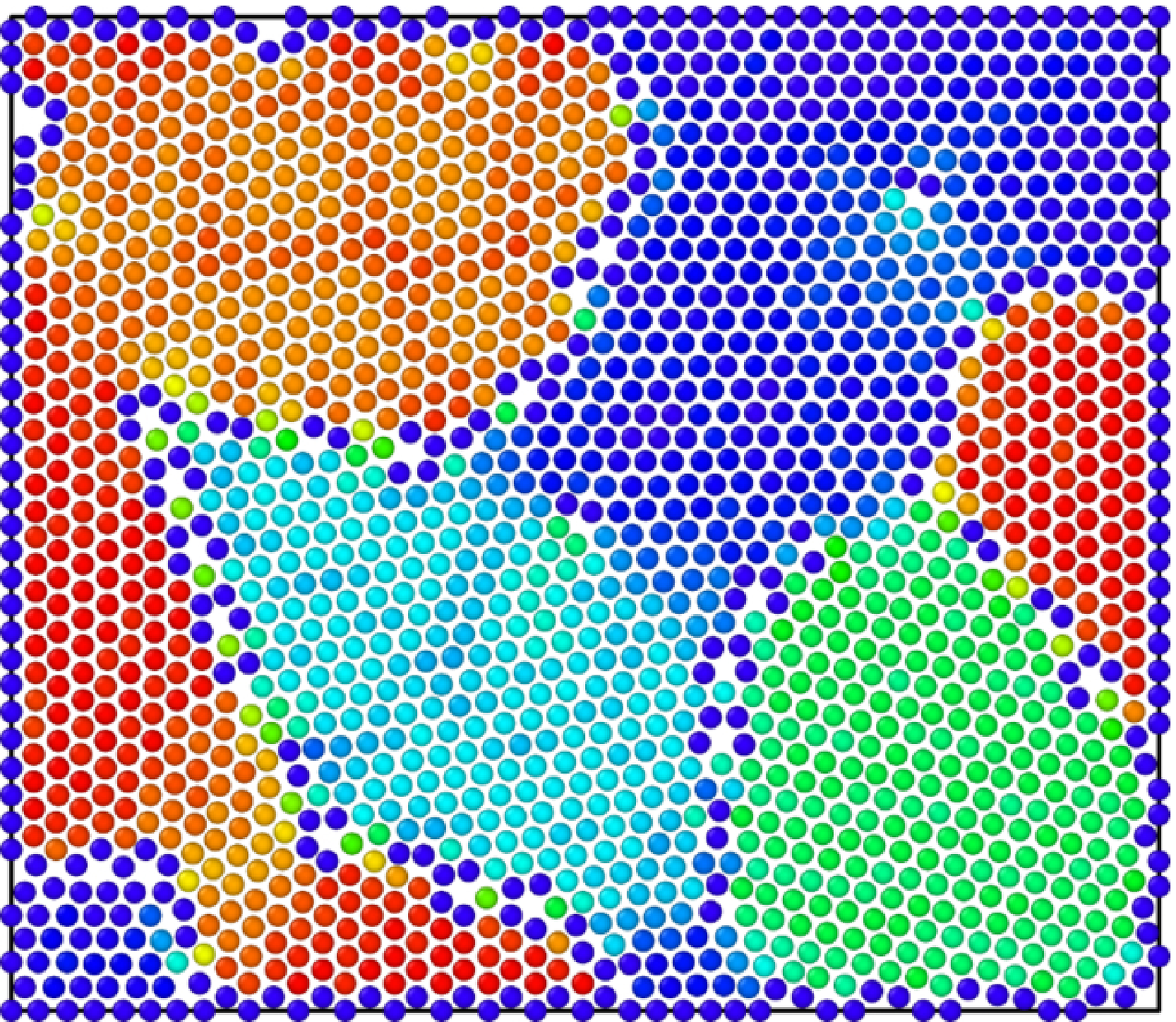} \\ 
\end{tabular} 
\caption{Grain growth simulation. Left: Initial grains that do not touch each other. Center: Grown grains with different orientation and grain boundaries. Right: Coloring of the different crystal orientations. The coloring fails at the boundary of the domain.}\label{fig:orientation}
\end{center}
\end{figure}

\begin{figure}[ht]
\begin{center}
\begin{tabular}{cc}
\includegraphics[scale=.96]{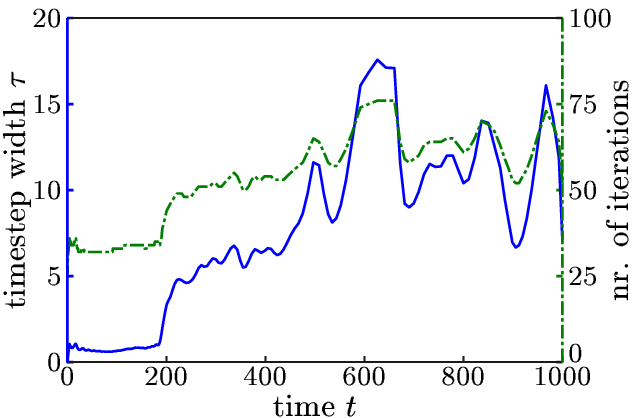} & 
\includegraphics[scale=.96]{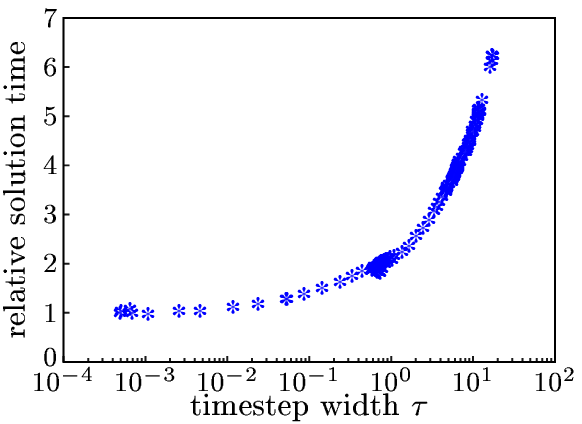}
\end{tabular} 
\caption{Left: Time series of timestep width (solid line) and outer solver iterations (dashed line) for a simulation using a Rosenbrock scheme with automatic step size selection. Right: Evolution of the solution time for increasing timestep widths. The time is measured relative to the time for the minimal timestep width. The data in extracted from the simulation of the grain growth, see Figure \ref{fig:orientation}.}\label{fig:timestep_solution_time}
\end{center}
\end{figure}

\subsection{Parallel calculations}
We now demonstrate parallel scaling properties. Figure \ref{fig:benchmark_speedup_Ts} shows strong and weak scaling results. All simulations are done in 3D and show results for the time to solve the linear system in comparison with a minimal number of processors that have the same communication and memory access environment. The efficiency of this strong scaling benchmark is about 0.8 -- 0.9 depending on the workload per processing unit. The efficiency of the weak scaling is about 0.9 -- 0.95, thus slightly better than the strong scaling. 

\begin{table}[ht]
\begin{center}
\begin{tabular}{c|c|c|c}
\#processors $p$ & total DOFs & time [sec] & \#iterations \\ 
\hline 
48 & 1,245,456 & 13.62 & 24 \\
96 & 2,477,280 & 13.57 & 24 \\
192 & 4,984,512 & 13.83 & 25 \\
384 & 9,976,704 & 14.97 & 25 \\
\end{tabular} 
\end{center}
\caption{Average number of iterations and solution time for weak scaling computations.}\label{tbl:iter_dofs}
\end{table}

In Table \ref{tbl:iter_dofs} the number of outer solver iterations for various system sizes is given. The calculations are performed in parallel on a mesh with constant grid size but with variable domain size. By increasing the number of processors respective the number of degrees of freedom in the system the number of solver iterations remains almost constant. Also the solution time changes only slightly.

Larger systems on up to 3.456 processors also show that the preconditioner does not perturb the scaling behavior of the iterative solvers. All parallel computations have been done on JUROPA at JSC.

\begin{figure}[ht]
\begin{center}
\begin{tabular}{cc}
\includegraphics[scale=.96]{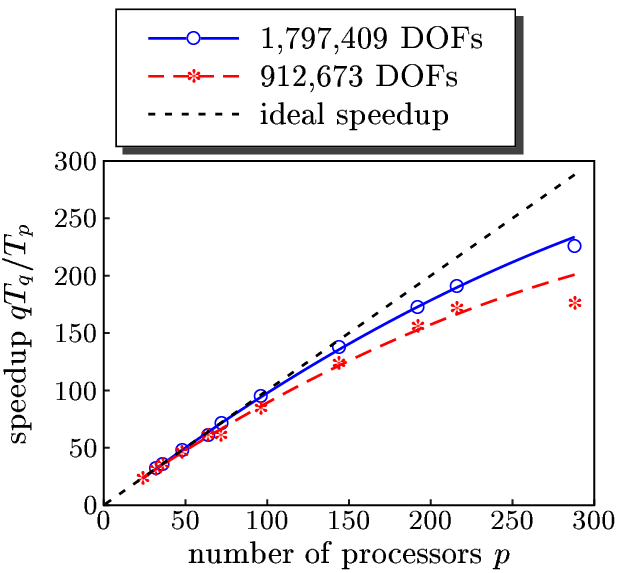} & 
\includegraphics[scale=.96]{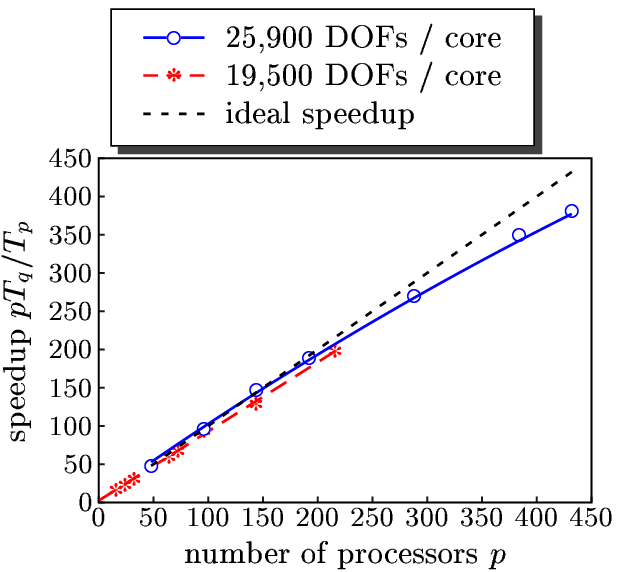}
\end{tabular} 
\caption{Speedup of parallel simulation. Left: strong scaling for fixed overall number of DOFs. Right: weak scaling for fixed number of DOFs per processor. The parameter $q$ is a reference number of processors, that is $q=24/32$ for the calculation in the left plot and $q=16/48$ for the right plot. The first number corresponds to the dashed lines and the second to the solid lines.}\label{fig:benchmark_speedup_Ts}
\end{center}
\end{figure}

\subsection{Non-regular domains}

While all previous benchmark problems could have also been simulated using spectral or finite difference methods, we now demonstrate two examples, where this is no longer possible and the advantage of the proposed solution approach is shown. We use parametric finite elements to solve the surface PFC equation \cite{Backofen2010} on a manifold. The first example considers an elastic instability of a growing crystal on a sphere, similar to the experimental results for colloidal crystals in \cite{Meng2014}. The observed branching of the crystal minimizes the curvature induced elastic energy, see Figure \ref{fig:sphere}. The second example shows a crystalline layer on a minimal surface, the `Schwarz P surface', Figure \ref{fig:schwarz_P}, which might be an approach to stabilize such surfaces by colloidal particles, see \cite{Irvine2010}.

\begin{figure}[ht]
\begin{center}
\begin{tabular}{cc}
\includegraphics[width=.28\linewidth]{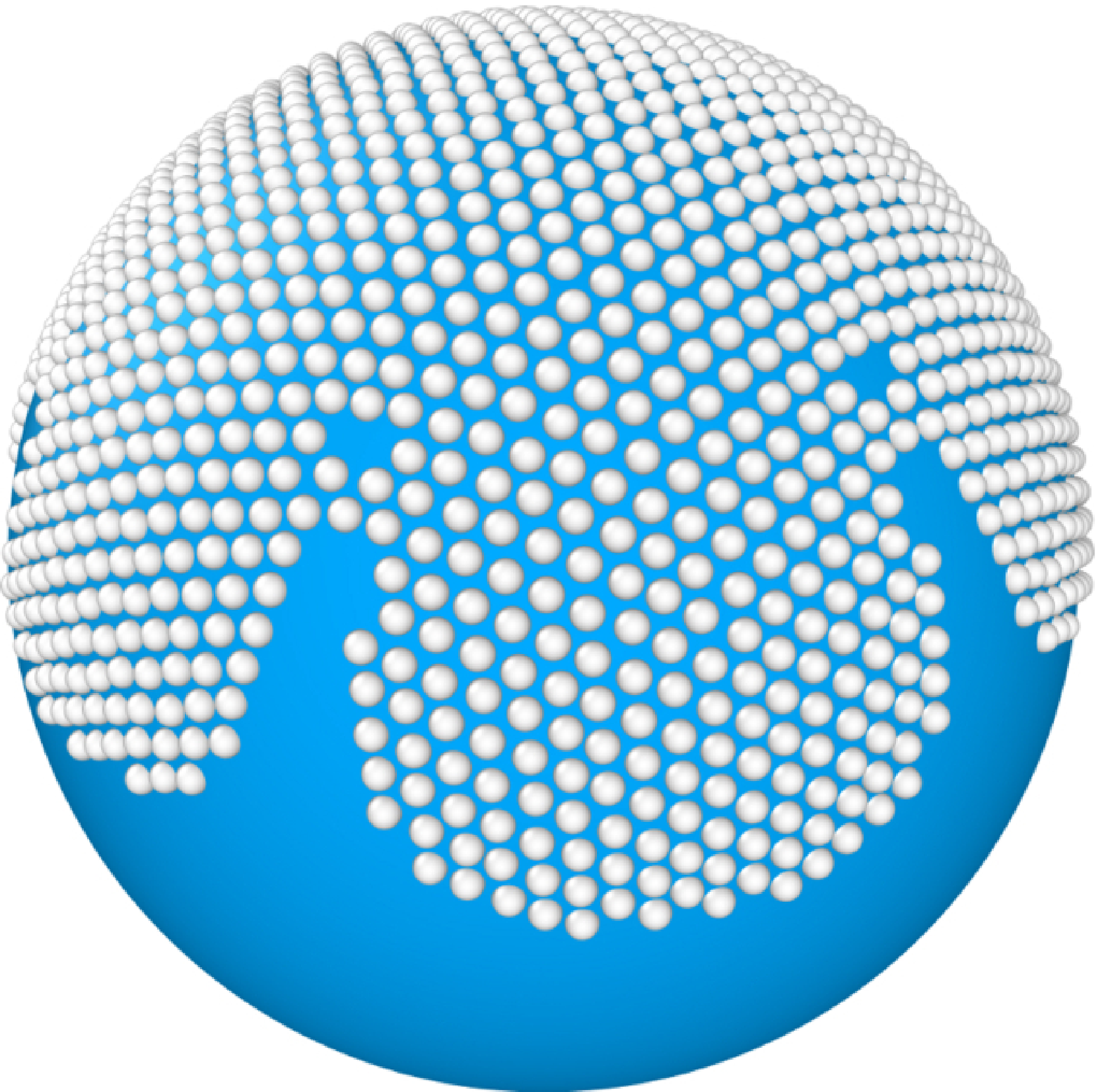}  &
\includegraphics[width=.28\linewidth]{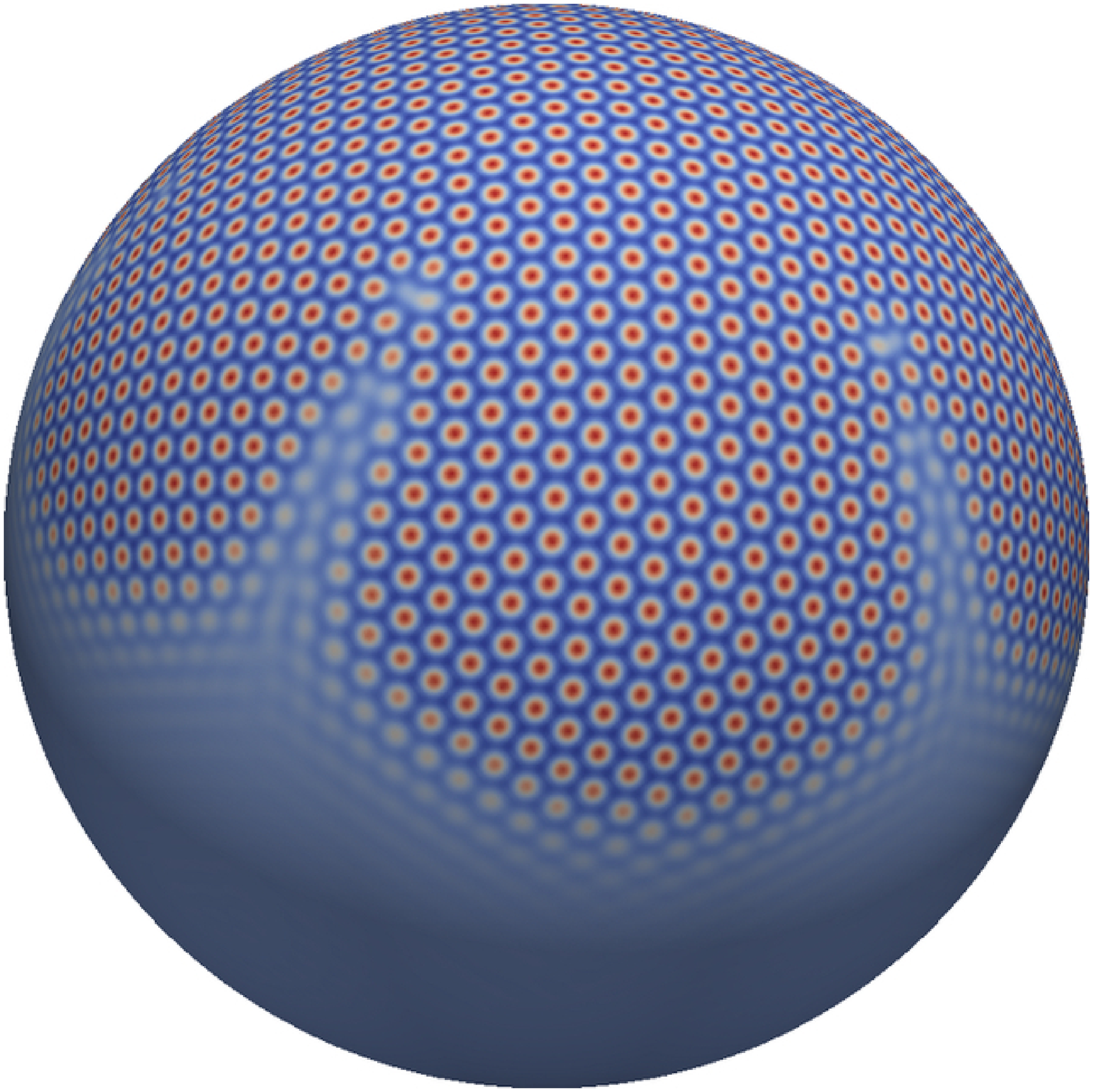}  
\end{tabular} 
\caption{Crystalization on a sphere $\mathcal{S}_{120}(0)$. Left: Visualization using OVITO \cite{Ovito}, indicating each wave as a colloidal particle, Right: order parameter field $\psi$. Number of DOFs: 397,584, calculated on 8 processors.} \label{fig:sphere} 
\end{center}
\end{figure}

\begin{figure}[ht]
\begin{center}
\begin{tabular}{ccc}
\includegraphics[width=.28\linewidth]{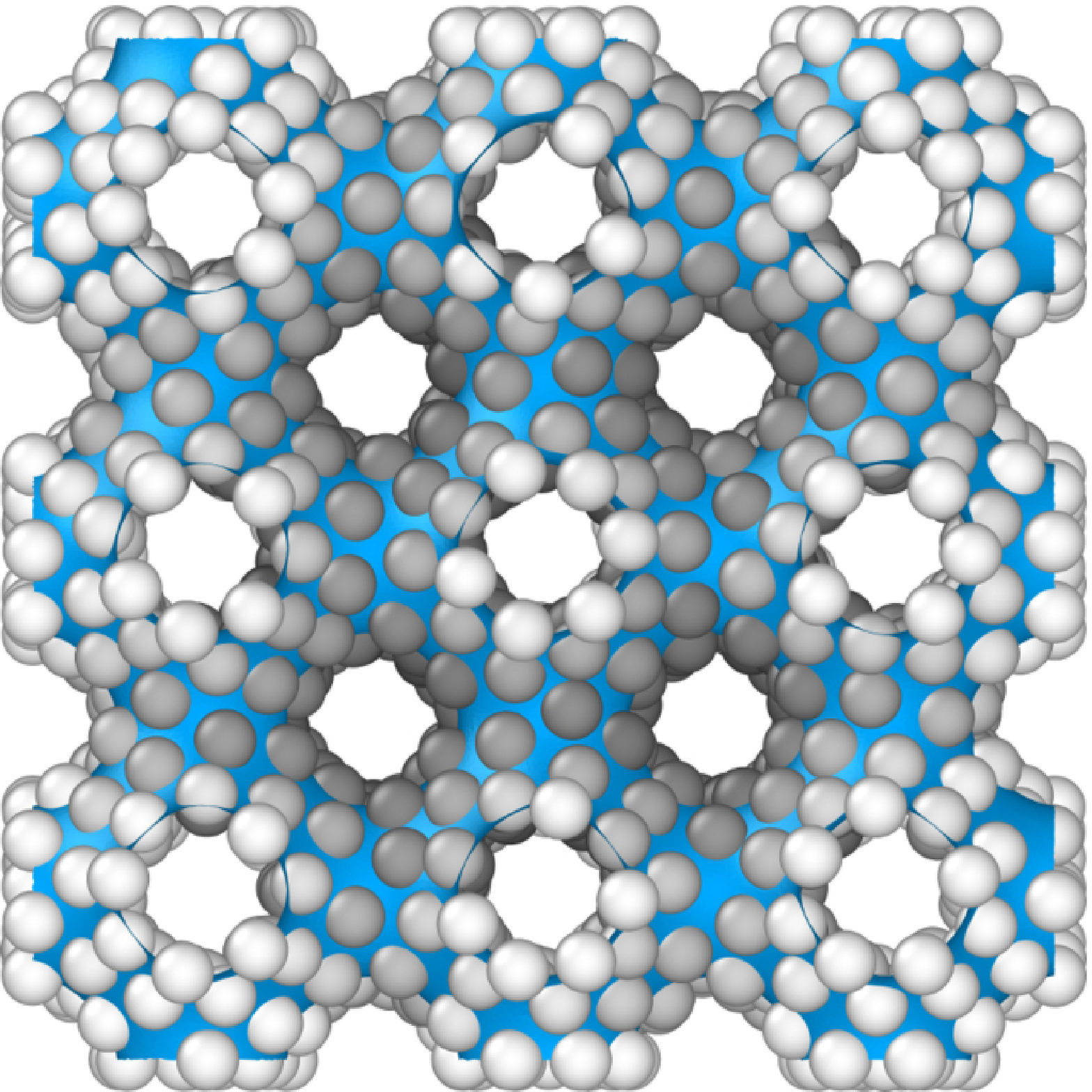}  &
\includegraphics[width=.28\linewidth]{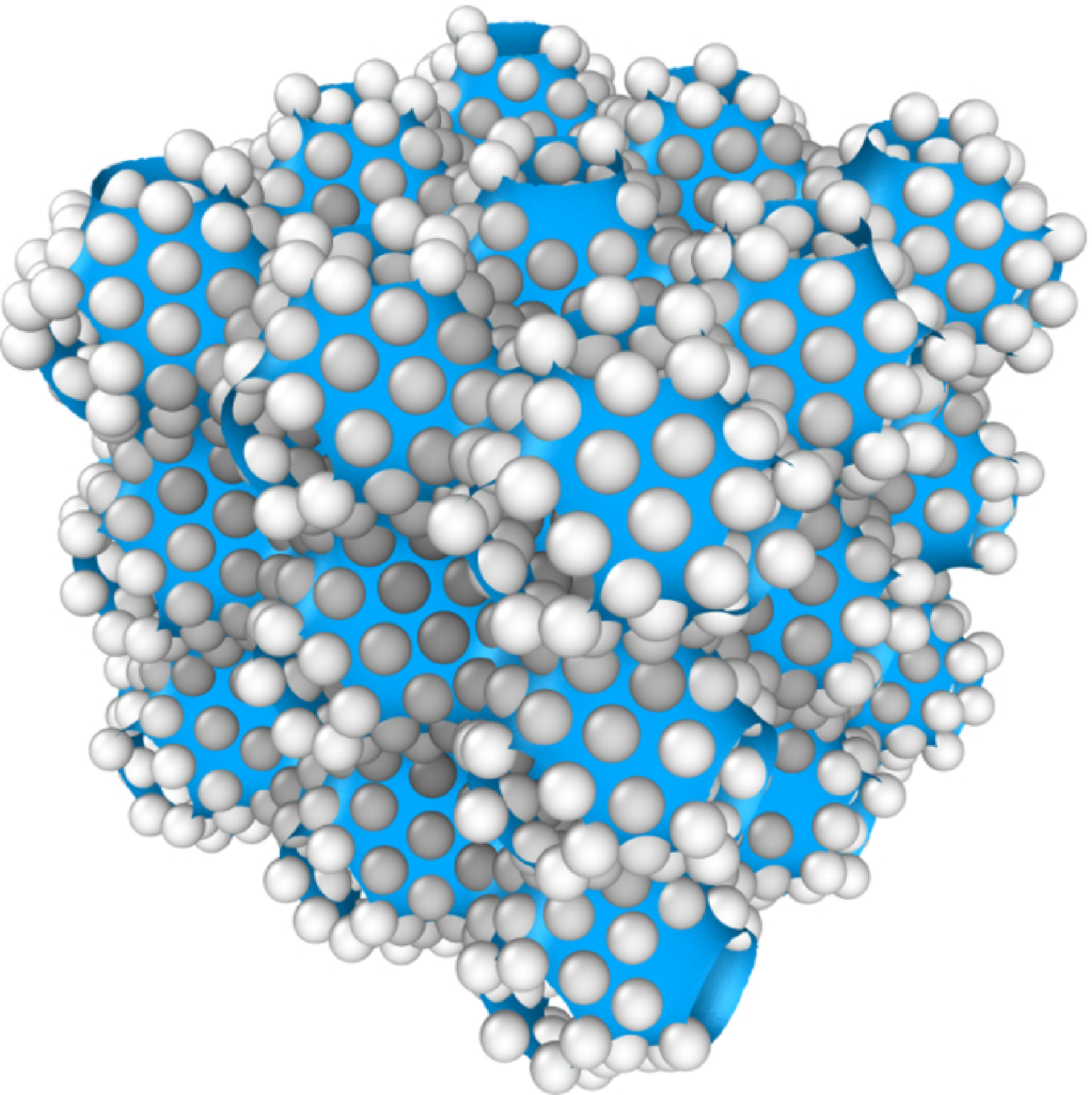}  &
\includegraphics[width=.28\linewidth]{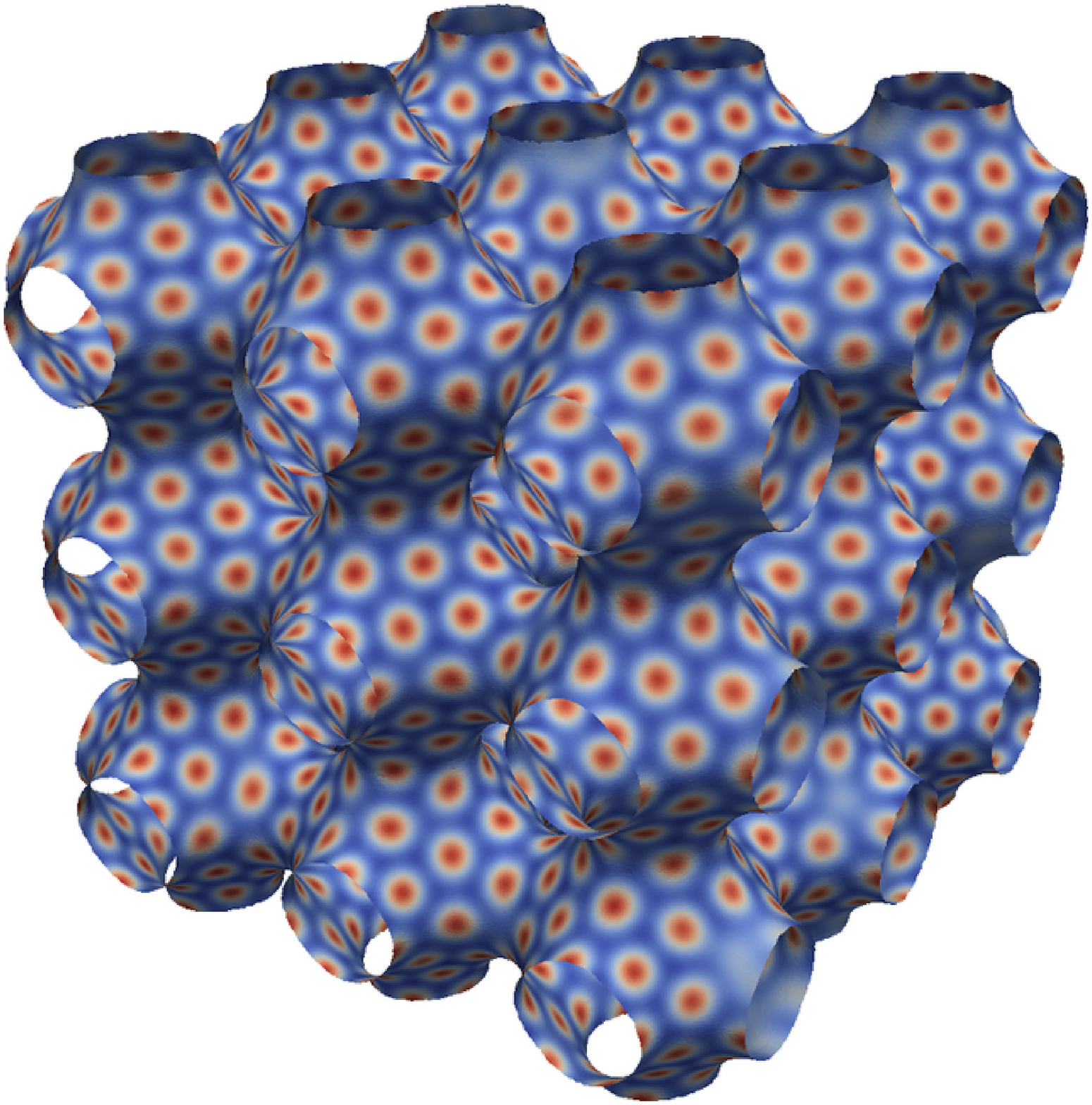} 
\end{tabular} 
\caption{Crystal structure on a  `Schwarz P surface'. Left and Center: Visualization using OVITO \cite{Ovito} in two different perspectives, indicating each wave as a colloidal particle, Right: order parameter field $\psi$. Number of DOFs: 250,000, calculated on 1 processor.} \label{fig:schwarz_P} %
\end{center}
\end{figure}

\section{Conclusion and outlook}\label{seq:conclusion}
In this paper we have developed a block-precon\-ditioner for the Phase-Field Crystal equation. It leads to a precondition procedure in 5 steps that can be implemented by composition of simple iterative solvers. Additionally, we have analyzed the preconditioner in Fourier-space and in numerical experiments. We have found a critical timestep width for the original preconditioner and have proposed a variant with an inner Cahn-Hilliard preconditioner, that does not show this timestep limit. Since most of the calculations are performed in parallel, a scaling study is provided, that shows, that there is no negative influence of the preconditioner on the scaling properties. Thus, large scale calculations in 2D and 3D can be performed.

Recently extensions of the classical PFC model are published, towards liquid crystalline phases \cite{Praetorius2013}, flowing crystals \cite{Menzel2013} and more. Analyzing the preconditioner for these systems, that are extended by additional coupling terms, is a planed task. 

Higher order models, to describe quasicrystalline states, respective polycrystalline states, based on a conserved Lifshitz-Petrich model \cite{Lifshitz1997}, respective a multimode PFC model \cite{Mkhonta2013}, are introduced and lead to even worse convergence behavior in finite element calculations than the classical PFC model. This leads to the question of an effective preconditioner for these models. The basic ideas, introduced here, might be applicable for the corresponding discretized equations as well.


\begin{appendix}

\end{appendix}

\bibliographystyle{siam}
\bibliography{references}
\end{document}